\documentclass[12pt,reqno]{amsart}

\allowdisplaybreaks[4]
\usepackage{graphicx} %We can use any other package if it is necessary
\usepackage{amssymb}
\usepackage{color}
\usepackage{amsmath}
\usepackage{cite}
\usepackage{subfigure}
\usepackage{graphicx}
\usepackage{epstopdf}
\usepackage{bibentry}
\usepackage{ulem}
\newtheorem{lemma}{Lemma}[section]

\newtheorem{thm}{Theorem}[section]
\begin{document}
\title{Higher-order rogue wave dynamics for a derivative nonlinear Schr\"odinger equation}
\author{Yongshuai Zhang$^1$, Lijuan Guo$^1$, Amin Chabchoub$^2$, Jingsong He$^{1*}$}
\thanks{$^*$ Corresponding author: hejingsong@nbu.edu.cn, jshe@ustc.edu.cn}
\dedicatory { $^1$ Department of Mathematics, Ningbo University,
Ningbo , Zhejiang 315211, P.\ R.\ China \\
$^2$Centre for Ocean Engineering Science and Technology, Swinburne University of Technology, Hawthorn, Victoria 3122, Australia
}

%%%%%%%%%%%%%%%%%%%%%%%%%%%%%%%%%%%%%%%%%%%%%%%%
\begin{abstract}
The the mixed Chen-Lee-Liu derivative nonlinear Schr\"odinger equation (CLL-NLS) can be considered as simplest model to approximate the dynamics of weakly nonlinear and dispersive waves, taking into account the self-steepnening effect (SSE). The latter effect arises as a higher-order correction of the nonlinear Schr\"ordinger equation (NLS), which is known to describe the dynamics of pulses in nonlinear fiber optics, and constiutes a fundamental part of the generalized NLS. Similar effects are decribed within the framework of the modified NLS, also referred to as the Dysthe equation, in hydrodynamics. In this work, we derive fundamental and higher-order solutions of the CLL-NLS by applying the Darboux transformation (DT). Exact expressions of non-vanishing boundary solitons, breathers and a hierarchy of rogue wave solutions are presented. In addition, we discuss the localization characters of such rogue waves, by characterizing their length and width. In particular, we describe how the localization properties of first-order NLS rogue waves can be modified by taking into account the SSE, presented in the CLL-NLS. This is illustrated by use of an analytical and a graphical method. The results may motivate similar analytical studies, extending the family of the reported rogue wave solutions as well as possible experiments in several nonlinear dispersive media, confirming these theoretical results.
\end{abstract}
%%%%%%%%%%%%%%%%%%%%%%%%%%%%%%%%%%%%%%%%%%%%%%%%

 \maketitle \vspace{-0.9cm}

\noindent {{\bf Keywords}: Chen-Lee-Liu derivative nonlinear Schr\"odinger equation, Darboux transformation, Rogue waves, Self-steepening effects}

\noindent {\bf PACS} numbers: 02.30.Ik,03.75.Lm,42.65.Tg\\
%02.30.Ik,integrable system;
%03.75.Lm, Tunneling, Josephson effect, BosegCEinstein condensates in periodic potentials,
        %solitons, vortices, and topological excitations
%42.65.Tg, Optical solitons; nonlinear guided waves

%%%%%%%%%%%%%%%%%%%%%%%%%%%%%%%%%%%%%%%%%%%%%%%%
\section{Introduction}
The nonlinear Schr\"odinger equation (NLS) is one of the most relevant equations in physics. This integrable equation can be rigorously derived as an approximation to governing equations of several nonlinear and dispersive media \cite{Benney1,Zakharov1,Hasegawa1,Ablowitz1}. Recently, a wide class of solutions, such as the Peregrine soliton \cite{Peregrine1} and multi-Peregrine soliton, also referred to as Akhmediev-Peregrine breathers \cite{Akhmediev1985}, of the NLS are intensively discussed in physical and mathematical communities \cite{OnoratoReport}. The doubly-localized Peregrine soliton, which approaches a non-zero constant background in the infinite limit of the spatial and temporal periodicity, amplifies the amplitude of the carrier by factor of three at the co-ordinates origin. Multi-Peregrine solitons \cite{Akhmediev1} have similar dynamics, with the particular property to generate much higher maximal peak amplitudes, compared to background
 \cite{Dubard1,Dubard2,Gaillard1,Ankiewicz1,Kedziora1,Ohta1,He1,Guo1}. Due to these properties, Peregrine-type waves are suggested to model ``rogue waves" (RWs), known to appear in the ocean \cite{Shrira} and in other media \cite{NatureReview}. Mathematically speaking, modulationally unstable extreme waves admit high-intensity peaks, appearing from nowhere and disappearing without a trace, while evolving in time and space \cite{Akhmediev2}. Recently, exact solutions of the NLS, describing a new form of modulation instability dynamics, have been derived \cite{ZakharovGelash,GelashZakharov}. The concept of the RWs was
  first discussed in the studies of ocean waves \cite{Pelinovshy1,KharifPelinovsky,KharifPelinovskySlunyaev,Osborne1}, and gradually extended to other fields of research, such as for instance for capillary
  water waves \cite{Shats1}, optical fibers \cite{Solli1,Solli2,Dudley1} and Bose-Einstein condensates \cite{Bludov1}, which have been 
  summarized in very recent review papers \cite{NatureReview, onorato}.

   Only recently, experimental validation of such RW model has been successfully conducted in nonlinear fibers \cite{Kibler1}, in water wave tanks \cite{Chabchoub1,Chabchoub2,Chabchoub3,Chabchoub4}, and in plasmas \cite{Bailung1,Sharma1}.
The latter experimental studies have been performed based on the NLS modeling evolution equation.

In addition to the NLS, there are several other integrable evolution equations admitting Peregrine-type RW solutions
such as the Hirota equation, the modified Korteweg-de Vries equation, the Sasa-Satsuma equation, the Fokas-Lenells
equation, the NLS Maxwell-Bloch equation, the Hirota Maxwell-Bloch equation, the generalized NLS, the vector NLS,
the derivative NLS, the variable coefficient NLS and derivative NLS, the Davey-Stewartson
equation, and the KP-I equation \cite{Ankiewicz2,He2,Bandelow1,Chenshihua,He3,He4,Li1,Zha1,
Wang1,Xu1,Xu2,Guoo1,Guo2,Zhang1,He5,Xu3,Ohta2,Ohta3,Dubard3,CPL1,PRE1,Degasperis1,Degasperis2,Degasperis3}. Lately, fundamental rogue wave modes of the mixed Chen-Lee-Liu derivative nonlinear Schr\"odinger equation (CLL-NLS) \cite{Kundu1}
\begin{equation}\label{MNLS}
  {\rm i}r_t+r_{xx}+|r|^2r-{\rm i}|r|^2r_x=0
\end{equation}
have been reported \cite{Chow1} by use of the Hirota bilinear method. Clearly, the latter solution is physically more complex and more accurate in describing the propagation of optical pulses compared to the NLS or simplified CLL Eq. \cite{Chen1}
\begin{equation}
  \label{cll}
  {\rm i}r_t+r_{xx}+{\rm i}|r|^2r_x=0,
\end{equation}
since the CLL-NLS takes into dispersion, nonlinearity as well as self-steepening effect (SSE), described by the term $|r|^2r_x$, however, while ignoring self-phase-modulation (SPM) \cite{Moses1}. The SSE of light pulses, originating from their propagation in a medium with an intensity dependent index of refraction, was first introduced in \cite{Demartini1} and was observed in optical pulses with possible shock formation
\cite{Grischkowsky1}. It receives a significant attention for the propagation of electromagnetic waves in nonlinear fibers, using a femtosecond laser, since it plays a crucial role in the generation of supercontinuum \cite{DudleyPhysicsToday,ChabchoubSupercontinuum}. In mathematical terms, its source is the first nonlinear correction to the NLS in the description of very focused light pulses or significant sharp water wave packets for which the validity of the NLS is known to be violated, due to the related significant broadening of the spectrum \cite{Tzoar1,Anderson1}. In hydrodynamics, the CLL-NLS can be obtained from the modified NLS, also known as the Dysthe equation \cite{Dysthe} by ignoring the mean flow term,  whose contribution is small if the nonlinearity of the wave train is kept small. Therefore, exact CLL-NLS models may motivate experiments in nonlinear optical fibers as well as in water wave flumes \cite{Chow1}. Especially, taking into account the fact that exact RW solutions are closely related to the modulation instability of weakly nonlinear dispersive waves.

In this paper, we report exact solutions of the integrable CLL-NLS. To the author's best knowledge, this is {}{so far} the first derivation of such doubly-localized solutions using the DT. In Section 2 and Section 3 the integration scheme will be introduced and we will address the significant challenges using the DT, solving CLL-type equations. These major difficulties are the result of the corresponding asymmetry of the Lax pair, see details in the appendix of \cite{Chow1}. Exact solutions with particular focus on higher-order RWs is reported in Section 4, extending therefore the family of exact first-order solutions. Furthermore, we discuss the influence of the SSE on the localization {}{properties} of NLS RWs in Section 5. Due to {}{obvious} physical relevance of the CLL-NLS, we emphasize further analytical, numerical and experimental studies, related to the presented exact solutions of this integrable evolution equation.

\section{The DT for the coupled CLL-NLS}
In this section, we consider the $n$-fold DT for the coupled CLL-NLS
\begin{equation}\label{cmnls}
  \left\{
  \begin{aligned}
  r_t-{\rm i}r_{xx}+{\rm i}r^2q+rqr_x=0,\\
  q_t+{\rm i}q_{xx}-{\rm i}q^2r+qrq_x=0,
  \end{aligned}\right.
\end{equation}
which reduces to the CLL-NLS while $q=-\overline{r}$ and the over-bar denotes complex conjugation.
These two equations in \eqref{cmnls} are the compatibility conditions of the following Lax pair \cite{Clarkson1,Lv}:
\begin{equation}\label{Lax}
\left\{
\begin{aligned}
  \Phi_x=U\Phi=&({\rm i}\sigma_3\lambda^2+Q\lambda-\frac{1}{2}{\rm i}\sigma_3+\frac{1}{4}{\rm i}Q^2\sigma_3)\Phi,\\
  \Phi_t=V\Phi=&[-2{\rm i}\sigma_3\lambda^4-2Q\lambda^3+\left(2{\rm i}\sigma_3-{\rm i}Q^2\sigma_3\right)
  \lambda^2+(Q+{\rm i}\sigma_3Q_x-\frac{1}{2}Q^3)\lambda\\
   &-\frac{1}{2}{\rm i}\sigma_3-\frac{1}{8}{\rm i}Q^4\sigma_3+\frac{1}{4}(QQ_x-Q_xQ)]\Phi,
\end{aligned}
\right.
\end{equation}
with
\begin{equation*}
  \Phi(x,t,\lambda)=\left(\begin{matrix}
    f(x,t,\lambda)\\
    g(x,t,\lambda)
  \end{matrix}\right),\quad \sigma_3=\left(\begin{matrix} 1 &0\\ 0 &-1 \end{matrix}\right),\quad
  Q=\left(\begin{matrix} 0 &r\\ q &0 \end{matrix}\right).
\end{equation*}
It is trivial to see that $\Phi_k\triangleq \left(  \begin{matrix}
f_k\\
g_k
 \end{matrix} \right) \triangleq  \Phi(x,t,\lambda)\left|_{\lambda=\lambda_k}\right.=
\left.\left(\begin{matrix}f(x,t,\lambda)\\g(x,t,\lambda)\end{matrix}\right)\right|_{\lambda=\lambda_k}
$ gives
the eigenfunction of the Lax pair equations corresponding to $\lambda_k$. Indeed, we seek $n$ eigenfunctions to
get the determinant representation of the $n$-fold DT.
%%%%%%%%%%%%%%%%%%%%%%%%%%%%%%%%%%%%%%%%%%%%%%%%%%%%%%%%%%%%%%%%%%%%%%%%%%%%%%%%%%%%
\begin{thm}\label{thm_nDT}
The $n$-fold DT for the coupled CLL-NLS is
  \begin{equation}
  \begin{aligned}
   T_n=T_n(\lambda;\lambda_1,\lambda_2,...,\lambda_n)=
   \begin{cases}
   \frac{1}{\sqrt{|\Delta_n^1||\Delta_n^2|}}
      \left(\begin{matrix}
      (T_n)_{11} &(T_n)_{12}\\
      (T_n)_{21} &(T_n)_{22}
      \end{matrix}\right) &\mbox{if $n$ is even},\\\\
   \frac{1}{\sqrt{|\Delta_n^1||\Delta_n^2|}}
     \left(\begin{matrix}
     \sqrt{H} &\\
     &\frac{1}{\sqrt{H}}
   \end{matrix}\right) \left(\begin{matrix}
      (T_n)_{11} &(T_n)_{12}\\
      (T_n)_{21} &(T_n)_{22}
      \end{matrix}\right) &\mbox{if $n$ is odd},
   \end{cases}
  \end{aligned}
  \end{equation}
%%%%%%%%%%%%%%%%%%%%%%%%%%%%%%%%%%%%%%%%
the elements $(T_n)_{ij}$ $(i,j=1,2)$  are defined by
    \begin{equation*}
    \begin{aligned}
    (T_n)_{11}=\begin{vmatrix}
      \lambda^n &\xi_n^1\\
      \eta_n^1 &\Delta_n^2
    \end{vmatrix},\quad
    (T_n)_{12}=\begin{vmatrix}
      0 &\xi_n^2\\
      \eta_n^1 &\Delta_n^2
    \end{vmatrix},\quad
    (T_n)_{21}=\begin{vmatrix}
      0 &\xi_n^2\\
     {\eta_n^2} &{\Delta_n^1}
    \end{vmatrix},\quad
    (T_n)_{22}=\begin{vmatrix}
      \lambda^n &\xi_n^1\\
      {\eta_n^2} &{\Delta_n^1}
    \end{vmatrix},
    \end{aligned}
  \end{equation*}
%%%%%%%%%%%%%%%%%%%%%%%%%%%%%%%%%%%%%%%%%%%%%%%%%%%%%
$\eta_n^i$, $\xi_n^i$ and $\Delta_n^i$ $(i=1,2)$ are defined by
\begin{equation*}
{\eta_n^1}=\left(\begin{matrix}
      \lambda_1^nf_1 &\lambda_2^nf_2  &\lambda_3^nf_3 &\ldots &\lambda_n^nf_n
    \end{matrix}\right)^T,\quad
    {\eta_n^2}=\left(\begin{matrix}
      \lambda_1^ng_1  &\lambda_2^ng_2 &\lambda_3^ng_3 &\ldots &\lambda_n^ng_n
    \end{matrix}\right)^T,
\end{equation*}
%%%%%%%%%%%%%%%%%%%%%%%%%%%%%%%%%%%%%%%%%%%%%%%%%%%%
 \begin{itemize}\setlength{\itemindent}{-2em}
    \item if $n$ is even,
    \begin{equation}
      \nonumber
      \xi_n^1=\left(\begin{matrix}
        0 &\lambda^{n-2} &0 &\lambda^{n-4} &\ldots &0 &1
      \end{matrix}\right),\quad
      \xi_n^2=\left(\begin{matrix}
        \lambda^{n-1} &0 &\lambda^{n-3} &0 &\ldots &\lambda &0
      \end{matrix}\right),
    \end{equation}
    \item if $n$ is odd,
    \begin{equation}
      \nonumber
      {\xi_n^1}=\left(\begin{matrix}
        0 &\lambda^{n-2} &0 &\lambda^{n-4} &\ldots &\lambda &0
      \end{matrix}\right),\quad
      {\xi_n^2}=\left(\begin{matrix}
        \lambda^{n-1} &0 &\lambda^{n-3} &0 &\ldots &0 &1
      \end{matrix}\right),
    \end{equation}
  \end{itemize}
%%%%%%%%%%%%%%%%%%%%%%%%%%%%%%%%%%%%%%%%%%%%%%%%%%%%%%%%%%%%%%%%%%
and
\begin{equation*}
  \Delta_n^1=\left(\begin{matrix}A_n^1 &A_n^2 &A_n^3 &\ldots &A_n^n\end{matrix}\right)^T,\quad
  \Delta_n^2=\left(\begin{matrix}B_n^1 &B_n^2 &B_n^3 &\ldots &B_n^n\end{matrix}\right)^T,
\end{equation*}
with $A_n^k,\,B_n^k$ $(k=1,2,3,\ldots,n)$ defined by
\begin{itemize}\setlength{\itemindent}{-2em}
    \item if $n$ is even,
    \begin{equation*}
      \begin{aligned}
      A_n^k=&\left(\begin{matrix} \lambda_k^{n-1}f_k &\lambda_k^{n-2}g_k &\lambda_k^{n-3}f_k &\lambda_k^{n-4}g_k &\ldots &\lambda_k^3f_k &\lambda_k^2g_k &\lambda_k^1f_k &g_k\end{matrix}\right),\\
      B_n^k=&\left(\begin{matrix} \lambda_k^{n-1}g_k &\lambda_k^{n-2}f_k &\lambda_k^{n-3}g_k &\lambda_k^{n-4}f_k &\ldots &\lambda_k^3g_k &\lambda_k^2f_k &\lambda_k^1g_k &f_k\end{matrix}\right),
      \end{aligned}
    \end{equation*}
    \item if $n$ is odd,
    \begin{equation*}
      \begin{aligned}
      A_n^k=&\left(\begin{matrix} \lambda_k^{n-1}f_k &\lambda_k^{n-2}g_k &\lambda_k^{n-3}f_k &\lambda_k^{n-4}g_k &\ldots &\lambda_k^3g_k &\lambda_k^2f_k &\lambda_k^1g_k &f_k\end{matrix}\right),\\
      B_n^k=&\left(\begin{matrix} \lambda_k^{n-1}g_k &\lambda_k^{n-2}f_k &\lambda_k^{n-3}g_k &\lambda_k^{n-4}f_k &\ldots &\lambda_k^3f_k &\lambda_k^2g_k &\lambda_k^1f_k &g_k\end{matrix}\right).
      \end{aligned}
    \end{equation*}
  \end{itemize}
\end{thm}
%%%%%%%%%%%%%%%%%%%%%%%%%%%%%%%%%%%%%%%%%%%%%%%%%%%%%%%%%%%%%%%%%%%%
 The solutions $(q_n,r_n)$ generated by the above n-fold DT have the following determinant representations.
\begin{thm}\label{thm_rn}
The $n$-th order solutions $r_n$ and $q_n$ are
\begin{equation}\label{qn}
   r_n=\begin{cases}
    \frac{|\Delta_n^1|}{|\Delta_n^2|}r-2{\rm i}\frac{|\Delta_n^4|}{|\Delta_n^2|} &\mbox{if $n$ is even},\\\\
    H\left(\frac{|\Delta_n^1|}{|\Delta_n^2|}r-2{\rm i}\frac{|\Delta_n^2|}{|\Delta_n^2|}\right) &\mbox{if $n$ is odd,}\end{cases}\qquad
   q_n=\begin{cases}
    \frac{|\Delta_n^2|}{|\Delta_n^1|}q-2{\rm i}\frac{|\Delta_n^3|}{|\Delta_n^1|} &\mbox{if $n$ is even},\\\\
    \frac 1H\left(\frac{|\Delta_n^2|}{|\Delta_n^1|}q-2{\rm i}\frac{|\Delta_n^3|}{|\Delta_n^1|}\right) &\mbox{if $n$ is odd},
    \end{cases}
  \end{equation}
the matrices $\Delta_n^j$ $(j=3,4)$ are defined by
\begin{equation*}
  \Delta_n^3=\left(\begin{matrix}C_n^1 &C_n^2 &C_n^3 &\ldots &C_n^n\end{matrix}\right)^T,\quad
  \Delta_n^4=\left(\begin{matrix}D_n^1 &D_n^2 &D_n^3 &\ldots &D_n^n\end{matrix}\right)^T,
\end{equation*}
%%%%%%%%%%%%%%%%%%%%%%%%%%%%%%%%%%%%%%%%%%%%%%%%%%%%%%%%%%%%%%%%%%%%%%%%%%%%%%%%%%%%%%%%%%%
with $C_n^k,\,D_n^k$ $(k=1,2,3,\ldots,n)$, given by
\begin{itemize}\setlength{\itemindent}{-2em}
    \item if $n$ is even,
    \begin{equation*}
      \begin{aligned}
      C_n^k=&\left(\begin{matrix} \lambda_k^{n}f_k &\lambda_k^{n-2}f_k &\lambda_k^{n-3}g_k &\lambda_k^{n-4}f_k &\ldots &\lambda_k^3g_k &\lambda_k^2f_k &\lambda_k^1g_k &f_k\end{matrix}\right),\\
      D_n^k=&\left(\begin{matrix} \lambda_k^{n}g_k &\lambda_k^{n-2}g_k &\lambda_k^{n-3}f_k &\lambda_k^{n-4}g_k &\ldots &\lambda_k^3f_k &\lambda_k^2g_k &\lambda_k^1f_k &g_k\end{matrix}\right),
      \end{aligned}
    \end{equation*}
    \item if $n$ is odd,
    \begin{equation*}
      \begin{aligned}
      C_n^k=&\left(\begin{matrix} \lambda_k^{n}f_k &\lambda_k^{n-2}f_k &\lambda_k^{n-3}g_k &\lambda_k^{n-4}f_k &\ldots &\lambda_k^3f_k &\lambda_k^2g_k &\lambda_k^1f_k &g_k\end{matrix}\right),\\
      D_n^k=&\left(\begin{matrix} \lambda_k^{n}g_k &\lambda_k^{n-2}g_k &\lambda_k^{n-3}f_k &\lambda_k^{n-4}g_k &\ldots &\lambda_k^3g_k &\lambda_k^2f_k &\lambda_k^1g_k &f_k\end{matrix}\right).
      \end{aligned}
    \end{equation*}
    \end{itemize}
\end{thm}

In theorem \ref{thm_nDT} and theorem \ref{thm_rn}, $(q,\,r)$ is a
 ``seed" solution of the coupled
CLL-NLS, $H$ is {}{an overall factor in the formula of the DT involved with} an integral function depending on $q$ and $r$, which satisfies the following conditions
\begin{equation}\label{Hc}
  \frac{\partial H}{\partial x}=\frac{1}{2}{\rm i}(qr-2)H,\quad \frac{\partial H}{\partial t}=-\frac{1}{4}(4{\rm i}+{\rm i}q^2r^2-2rq_x+2qr_x)H.
\end{equation}
A general analytical expression of $H$  is
\begin{equation}\label{HCLLNLS}
  H=\exp\left(\int^{(x,\,t)}_{(x_0,\,t_0)}\frac{1}{2}\mathrm{i}
  (qr-2)\mathrm{d}x-\frac{1}{4}(4{\rm i}+{\rm i}q^2r^2-2rq_x+2qr_x)\mathrm{d}t\right).
\end{equation}
{}{Let $a,c $ be two real constants, $b=a^2 +(a-1)c^2$, and then $q=-\overline{r}=c\exp\left({\rm i}(ax+bt)\right)$
 is a ``seed" solution of the CLL-NLS. For this case,}
\begin{equation}\label{HHE}
 H=\exp(-\frac 12\,{\rm i}( 2+{c}^{2}) x-\frac 14{\rm i}(4+\,{c}^{4}+4{c}^{2}a)t),
\end{equation}
{}{which will be used to generate breather solution  of the CLL-NLS by DT later.}
\section{Derivation of the $n$-fold DT}
In this section, we derive the $n$-fold DT and the $n$-th order solutions for the coupled
CLL-NLS {}{in order to prove  theorem \ref{thm_nDT} and theorem \ref{thm_rn}}. To obtain
the $n$-fold DT we consider the one- and two-fold DT at first, and then the $n$-fold DT can be obtained by iteration.

\subsection{The one-fold DT}
Without loss of generality, assuming the one-fold DT as
\begin{equation}
  \label{DT1}
  T_1(\lambda)=\left(\begin{matrix} a_1 &b_1\\ c_1 &d_1\end{matrix}\right)\lambda+\left(\begin{matrix}
  a_0 &b_0\\ c_0 &d_0\end{matrix}\right),
\end{equation}
$a_k$, $b_k$, $c_k$ and $d_k$ $(k=0,1)$ are complex functions of $x$ and $t$. Then, there
exists $\Phi^{[1]}=T_1\Phi$ satisfying the following conditions $\Phi^{[1]}_x=U^{[1]}\Phi^{[1]}$
and $\Phi^{[1]}_t=V^{[1]}\Phi^{[1]}$, where $U^{[1]}$ and $V^{[1]}$ have the same form as $U$ and $V$
except that $q$ and $r$ are replaced by $q_1$ and $r_1$. If so, we have
\begin{equation}\label{DTc}
  T_x+TU-U^{[1]}T=0, \quad\mbox{ and }\quad T_t+TV-V^{[1]}T=0.
\end{equation}
%%%%%%%%%%%%%%%%%%%%%%%%%%%%%%%%%%%%%%%%%%%%%%%%%%%%%%%%%%%%%%%%%%%%%%%%%%%%
\begin{lemma} Let one-fold DT of the coupled CLL-NLS  be the form  of (\ref{DT1}), then it is given by
  \begin{equation}\label{nDT_1}
    T_1(\lambda)=T_1(\lambda,\lambda_1)=\frac{1}{\sqrt{f_1g_1}}\left(\begin{matrix}\sqrt{H} & \\ &\frac{1}{\sqrt{H}}
    \end{matrix}\right)\left(\begin{matrix} \lambda g_1 &-\lambda_1f_1\\ -\lambda_1g_1 &\lambda f_1\end{matrix}\right),
  \end{equation}
  and the new solution $(q_1,\,r_1)$, generated by above $T_1 $ from ``seed" $(q,r)$ is
  \begin{equation}\label{rn_1}
    r_1=H\left(\frac{g_1}{f_1}r+2{\rm i}\lambda_1\right),\quad q_1=\frac{1}{H}\left(\frac{f_1}{g_1}q-2{\rm i}\lambda_1\right).
  \end{equation}
  {}{Here, the overall factor $H$ is given by (\ref{HCLLNLS}).}
\end{lemma}
\begin{proof}
  Let $F(\lambda)=(F_{ij})=T_x+TU-U^{[1]}T=0$ $(i,j=1,2)$ and substitute $T_1$ \eqref{DT1} into $F$, then
  \begin{equation*}
  \begin{aligned}
    F_{11}=&(qb_1-r_1c_1)\lambda^2+\left(qb_0-r_1c_0+a_{1x}+\frac 14{\rm i}a_1\left(qr-q_1r_1\right)\right)\lambda+a_{0x}+\frac 14{\rm i}a_0\left(qr-q_1r_1\right),\\
    F_{12}=&-2{\rm i}\lambda^3b_1+\left(ra_1-r_1d_1-2{\rm i}b_0\right)\lambda^2+\left(ra_0-r_1d_0+{\rm i}b_1+b_{1x}-\frac 14{\rm i}b_1\left(qr+q_1r_1\right)\right)\lambda\\
           &+{\rm i}b_0+b_{0x}-\frac 14{\rm i}b_0\left(qr+q_1r_1\right),\\
    F_{21}=&2{\rm i}\lambda^3c_1+\left(qd_1-q_1a_1+2{\rm i}c_0\right)\lambda^2+\left(qd_0-q_1a_0-{\rm i}c_1+c_{1x}+\frac 14{\rm i}c_1\left(qr+q_1r_1\right)\right)\lambda\\
           &-{\rm i}c_0+c_{0x}+\frac 14{\rm i}c_0\left(qr+q_1r_1\right),\\
    F_{22}=&(rc_1-q_1b_1)\lambda^2+\left(rc_0-q_1b_0+d_{1x}-\frac 14{\rm i}d_1\left(qr-q_1r_1\right)\right)\lambda+d_{0x}-\frac 14{\rm i}d_0(qr-q_1r_1).
  \end{aligned}
  \end{equation*}
  {}{ Note  that  $b_1$ and $c_1$ are equal to zero from coefficient of $\lambda^3$, and then remaining coefficients of $\lambda^i(i=0,1,2)$ imply}
  \begin{equation}\label{rr}
    r_1=\frac {a_1}{d_1} r-\frac {2{\rm i}b_0}{d_1},\quad q_1=\frac {d_1}{a_1} q+\frac {2 {\rm i}c_0}{a_1},
  \end{equation}
  and
  \begin{equation}
  \begin{split}
    a_{1x}&=\frac{a_1c_0}{2d_1}r-\frac{{\rm i}b_0c_0}{d_1}-\frac{b_0}{2}q,\quad
    d_{1x}=\frac{d_1b_0}{2a_1}q+\frac{{\rm i}b_0c_0}{a_1}-\frac{c_0}{2}r,\\
    b_{0x}&=\frac{b_0^2}{2a_1}q-\frac{b_0c_0}{2d_1}r+\frac{{\rm i}b_0^2c_0}{a_1d_1}+\frac 12{\rm i}b_0qr-{\rm i}b_0,\quad
    c_{0x}=\frac{c_0^2}{2d_1}r-\frac{b_0c_0}{2a_1}q-\frac{{\rm i}c_0^2b_0}{a_1d_1}-\frac 12{\rm i}c_0qr+{\rm i}c_0.
  \end{split}
  \end{equation}
  Let $a_0=d_0=0$ according to the coefficients of $\lambda$ in order
  to obtain the non-trivial solution. After simple calculations, we obtain $(a_1d_1)_x=0$,
   $(b_0c_0)_x=0$ and $(a_1b_0)_x=\frac 12 {\rm i}a_1b_0(qr-2)$.
  Based on the above results and taking the similar procedure to the second
  {}{formula of}  \eqref{DTc}, we have $(a_1d_1)_t=0$, $(b_0c_0)_t=0$ and  $(a_1b_0)_t=-\frac 14
   a_1b_0(4{\rm i}+{\rm i}q^2r^2-2rq_x+2qr_x)$. Now, let $a_1d_1=1$ and $b_0c_0=\lambda_1^2$ without loss of generality. Moreover, according to $(a_1b_0)_{xt}=(a_1b_0)_{tx}$, it is reasonable to let $a_1b_0=\mu G$,
  where $G$ is the primitive integral function and $\mu$ is an integral constant. That is, $G$ satisfies
  \begin{equation}
    \frac{\partial G}{\partial x}=\frac{1}{2}{\rm i}(qr-2)G,\quad \frac{\partial G}{\partial t}=-\frac{1}{4}(4{\rm i}+{\rm i}q^2r^2-2rq_x+2qr_x)G.
  \end{equation}
  Thus, $G=H$, if we disregard the integral constant.
%%%%%%%%%%%%%%%%%%%%%%%%%%%%%%%%%%%%%%%%%%%%%%%%%%%%%%%%%%%%%%%%%%%%%%%%%%%%%%%%%%%%%%%%%%%%%%%%%%

  {}{The explicit form of $T_1$ can be determined by $T_1\Phi_1|_{\lambda=\lambda_1}=0$, i.e.}
  \begin{equation*}
    a_1\lambda_1f_1+b_0g_1=0,\quad c_0f_1+d_1\lambda_1g_1=0.
  \end{equation*}
  For convenience, let $\mu=-\lambda_1$, then {}{unknown elements}  $a_1$, $d_1$, $b_0$, and $c_0$ are  solved by
  \begin{equation*}
    a_1=\sqrt{H}\sqrt{\frac{g_1}{f_1}},\quad d_1=\frac{1}{\sqrt{H}}\sqrt{\frac{f_1}{g_1}},
    \quad b_0=-\lambda_1\sqrt{H}\sqrt{\frac{f_1}{g_1}},\quad c_0=-\lambda_1\frac{1}{\sqrt{H}}\sqrt{\frac{g_1}{f_1}}.
  \end{equation*}
  That is, the form of one-fold DT is
  \begin{equation*}
     T_1(\lambda)=T_1(\lambda,\lambda_1)=\left(\begin{matrix} \lambda \sqrt{H}\sqrt{\frac{g_1}{f_1}}
     &-\lambda_1\sqrt{H}\sqrt{\frac{f_1}{g_1}}\\
     -\lambda_1\frac{1}{\sqrt{H}}\sqrt{\frac{g_1}{f_1}} & \lambda \frac{1}{\sqrt{H}}\sqrt{\frac{f_1}{g_1}}  \end{matrix}\right),
  \end{equation*}
  and the new solution $(q_1,\,r_1)$ can be expressed as
  \begin{equation*}
   r_1=H\left(\frac{g_1}{f_1}r+2{\rm i}\lambda_1\right),\quad q_1=\frac{1}{H}
   \left(\frac{f_1}{g_1}q-2{\rm i}\lambda_1\right).
  \end{equation*}
\begin{flushright}
Q.E.D.
\end{flushright}
\end{proof}
{}{
Note that transformed eigenfunctions associated with new solution $(q_1,r_1)$ are
\begin{equation}\label{thefirstDTegien}
\begin{matrix}
\Phi_j^{[1]}=\left(\begin{matrix} f_j^{[1]}\\ g_j^{[1]}\end{matrix}\right)=T(\lambda,\lambda_1)|_{\lambda=\lambda_j}\Phi_j.
\end{matrix}
\end{equation}It is trivial to see $\Phi^{[1]}_1=0$. In other words, $T_1$ annihilates its generating function which is a general property of the DT.
Therefore, we have to use a transformed eigenfunction $\Phi_2^{[1]}$ associated with $\lambda_2(\not= \lambda_1)$ in order to generate the next step DT.
}

\subsection{The two-fold DT}
By iteration, the two-fold DT for the coupled CLL-NLS is calculated as
\begin{equation*}
  T_2(\lambda)=T_2(\lambda,\lambda_1,\lambda_2)=T_1^{[1]}(\lambda,\lambda_2)T_1(\lambda,\lambda_1),
\end{equation*}
where
\begin{equation*}
  T_1^{[1]}(\lambda,\lambda_2)=\frac{1}{\sqrt{f_2^{[1]}g_2^{[1]}}}\left(\begin{matrix}\sqrt{H_1}
  & \\ &\frac{1}{\sqrt{H_1}}
    \end{matrix}\right)\left(\begin{matrix} \lambda g_2^{[1]} &-\lambda_2f_2^{[1]}\\
    -\lambda_2g_2^{[1]} &\lambda f_2^{[1]}\end{matrix}\right),
\end{equation*}
$H_1$ possesses the same form as $H$ in (\ref{HCLLNLS}), except  $q$ and $r$ replaced by $q_1$ and $r_1$. The definitions of $H_1$ and $\Phi_2^{[1]}$ are valid for $H_k$
and $\Phi_k^{[j]}$ (If $k<j$, $\Phi_k^{[j]}=0$)).  {}{According to the specific matrix forms of $T_1$ and $T_1^{[1]}(\lambda,\lambda_2)$, then
$T_2$  is expressed by }

\begin{equation}
  T_2(\lambda;\ \lambda_1,\lambda_2)=\left(\begin{matrix}
    a_2^{[1]}&\\
    &d_2^{[1]}
  \end{matrix}\right)\lambda^2+\left(\begin{matrix}
    &b_1^{[1]}\\
    c_1^{[1]}&
  \end{matrix}\right)\lambda+\left(\begin{matrix}
    a_0^{[1]}& \\ &d_0^{[1]}
  \end{matrix}\right),
\end{equation}
and
\begin{equation*}
  a_0^{[1]}=\lambda_1\lambda_2\sqrt{\frac{H_1f_2^{[1]}g_1}{Hg_2^{[1]}f_1}}, \quad
  d_0^{[1]}=\lambda_1\lambda_2\sqrt{\frac{Hg_2^{[1]}f_1}{H_1f_2^{[1]}g_1}}.
\end{equation*}
Note that $T_2(\lambda)\Phi_k|_{\lambda=\lambda_k}=0$ $(k=1,2)$, then four unknown elements  $a_2^{[1]}$, $d_2^{[1]}$, $b_1^{[1]}$, $c_1^{[1]}$ can be solved as follows according to Cramer's rule,
\begin{equation*}
  a_2^{[1]}=\frac{\delta_3}{\delta_1},\quad b_1^{[1]}=\frac{\delta_5}{\delta_1},\quad d_2^{[1]}=\frac{\delta_4}{\delta_2}, \quad c_1^{[1]}=\frac{\delta_6}{\delta_2},
\end{equation*}
where $\delta_k$ $(k=1,2,\ldots,6)$ are defined by
\begin{equation*}
 \begin{aligned}
 \delta_1=\begin{vmatrix} \lambda_1^2f_1 &\lambda_1g_1\\ \lambda_2^2f_2 &\lambda_2g_2 \end{vmatrix},  \quad
 \delta_2=\begin{vmatrix} \lambda_2^2g_1 &\lambda_1f_1\\ \lambda_2^2g_2 &\lambda_2f_2 \end{vmatrix},  \quad
 \delta_3=\begin{vmatrix} -a_0^{[1]}f_1 &\lambda_1g_1\\ -a_0^{[1]}f_2 &\lambda_2g_2 \end{vmatrix},\\
 \delta_4=\begin{vmatrix} -d_0^{[1]}g_1 &\lambda_1f_1\\ -d_0^{[1]}g_2 &\lambda_2f_2 \end{vmatrix},    \quad
 \delta_5=\begin{vmatrix} \lambda_1^2f_1 &-a_0^{[1]}f_1\\ \lambda_2^2f_2 &-a_0^{[1]}f_2 \end{vmatrix},\quad
 \delta_6=\begin{vmatrix} \lambda_1^2g_1 &-d_0^{[1]}g_1\\ \lambda_2^2g_2 &-d_0^{[1]}g_2 \end{vmatrix}.
 \end{aligned}
\end{equation*}
%%%%%%%%%%%%%%%%%%%%%%%%%%%%%%%%%%%%%%%%%%%%%%%%%%%%%%%%%%%%%%%%%%%%%%%%%%%%%%%%%%%%%%%%%%%%%%%%%%%%%%%%%%%%%%%
{}{Substituting above elements in matrix form of $T_2$,  then it becomes }
\begin{equation}\label{DT2}
  T_2(\lambda)=T_2(\lambda,\lambda_1,\lambda_2)=\frac{1}{\sqrt{\begin{vmatrix} \lambda_1f_1 &g_1\\ \lambda_2f_2 &g_2\end{vmatrix}\begin{vmatrix} \lambda_1g_1 &f_1\\ \lambda_2g_2 &f_2 \end{vmatrix}}}\left(\begin{matrix} \sqrt{\frac{g_1H_1}{f_1}} &\\ &\sqrt{\frac{f_1}{g_1H_1}}\end{matrix}\right)\left(\begin{matrix} (T_2)_{11} &(T_2)_{12}\\
  (T_2)_{21} &(T_2)_{22}\end{matrix}\right),
\end{equation}
and
 elements $(T_2)_{ij}$ $(i,j=1,2)$ are given by following determinants
\begin{equation}\nonumber
\begin{aligned}
  (T_2)_{11}=\begin{vmatrix} \lambda^2 &0 &1 \\ \lambda_1^2f_1 &\lambda_1g_1 &f_1\\ \lambda_2^2f_2 &\lambda_2g_2 &f_2 \end{vmatrix},\quad
  (T_2)_{12}=\begin{vmatrix} 0 &\lambda &0\\ \lambda_1^2f_1 &\lambda_1g_1 &f_1\\ \lambda_2^2f_2 &\lambda_2g_2 &f_2\end{vmatrix},\\
  (T_2)_{21}=\begin{vmatrix} {}{0} & \lambda &  0 \\ \lambda_1^2g_1 &\lambda_1f_1 &g_1\\ \lambda_2^2g_2 &\lambda_2f_2 &g_2 \end{vmatrix},\quad
  (T_2)_{22}=\begin{vmatrix} {} {\lambda^2} & 0 &1 \\ \lambda_1^2g_1 &\lambda_1f_1 &g_1\\ \lambda_2^2g_2 &\lambda_2f_2 &g_2\end{vmatrix}.
\end{aligned}
\end{equation}
Note that {}{the overall factor $H_1$ has an integral function} depending on $q_1$ and $r_1$. It implies
 that we need {}{to} apply the one-fold DT in order to obtain the two-fold. Thus, $T_2$ is not an explicit
 formula of the two-fold DT. Especially, as one iterates the above method,
 more integrals in overall factors  $H_k$ $(k>1)$ will be involved. This depends on $q_k$ and $r_k$.
 However, $q_k$ and $r_k$ are too cumbersome to be expressed in terms of explicit integrals in overall factors $H_k$. That is, it is not possible to get the explicit expressions of $T_k$ {} {if  $H_k$ can not be eliminated}.  Thus, eliminating
 the integrals in the overall factors $H_k$ is an unavoidable challenge. The next Lemma provides a crucial step to
 deal with this obstacle.  In the following lemma, $\dfrac{g_i^{[0]}}{f_i^{[0]}} \triangleq \dfrac{g_i}{f_i}$.
\begin{lemma}\label{lemma_H}
  Let {}{$i\geq k+1\geq1$}, then $\frac{g_i^{[k]}}{f_i^{[k]}}H_{k+1}$ is a constant.
\end{lemma}
%%%%%%%%%%%%%%%%%%%%%%%%%%%%%%%%%%%%%%%%%%%%%%%%%%%%%%%%%%%%%%%%%%%%%%%%%%%%%%%%%%%%%%%%%%%%%%%%
\begin{proof}
  On one hand, {}{according to the x-part of the Lax pair for $\Phi^{[k]}_i$ and the $k$-th step of DT, a straightforward calculation implies}
  \begin{equation}\nonumber
    f_{ix}^{[k]}=({\rm i}\lambda_i^2-\frac 12 {\rm i}+\frac 14 {\rm i}q_kr_k)f_i^{[k]}+\lambda_ir_kg_i^{[k]},\quad g_{ix}^{[k]}=\lambda_iq_kf_i^{[k]}-({\rm i}\lambda_i^2-\frac 12 {\rm i}+\frac 14 {\rm i}q_kr_k)g_i^{[k]},
  \end{equation}
  \begin{equation}\nonumber
    r_{k+1}=H_k(\frac{g_i^{[k]}}{f_i^{[k]}}r_k+2{\rm i}\lambda_i),\quad q_{k+1}=\frac{1}{H_k}(\frac{f_i^{[k]}}{g_i^{[k]}}q_k-2{\rm i}\lambda_i).
  \end{equation}
According to the definition of $H_{k+1}$,
  $$
  \frac{H_{k+1,x}}{H_{k+1}}=\frac 12{\rm i}(q_{k+1}r_{k+1}-2)=\frac 12{\rm i}q_kr_k-{\rm i}+2{\rm i}\lambda_i^2-\lambda_i
  (\frac{f_i^{[k]}}{g_i^{[k]}}q_k-\frac{g_i^{[k]}}{f_i^{[k]}}r_k).
  $$
  Thus,
  \begin{equation}
  \begin{aligned}
    (\frac{g_i^{[k]}}{f_i^{[k]}}H_{k+1})_x=\frac{g_{ix}^{[k]}}{g_{i}^{[k]}}-\frac{f_{ix}^{[k]}}{f_{i}^{[k]}}
    +\frac{H_{k+1,x}}{H_{k+1}}=0.
  \end{aligned}
  \end{equation}
%%%%%%%%%%%%%%%%%%%%%%%%%%%%%%%%%%%%%%%%%%%%%%%%%%%%%%%%%%%%%%%%%%%%%%%%%%%%%%%%%%%%%%%%%%%%%%%%%%%%%%%%%%
%%%%%%%%%%%%%%%%%%%%%%%%%%%%%%%%%%%%%%%%%%%%%%%%%%%%%%%%%%%%%%%%%%%%%%%%%%%%%%%%%%%%%%%%%%%%%%%%%%%%%%%%%%
On the other hand, according to the $t$-part of Lax pair for $\Phi^{[k]}_i$, and the definition of $H_{k+1,t}$, a straightforward
 calculation implies
\begin{equation}\nonumber
\begin{aligned}
  \frac{f_{it}^{[k]}}{f_i^{[k]}}=&-2{\rm i}\lambda_i^4+\left(2{\rm i}-{\rm i}q_kr_k\right)\lambda_i^2-\frac 18 {\rm i}q_k^2r_k^2+\frac 14 r_kq_{k,x}-\frac 14 q_kr_{k,x}-\frac 12{\rm i}\\
  &-(2r_k\lambda_i^3-(r_k-\frac 12 r_k^2q_k+{\rm i}r_{k,x})\lambda_i) \frac{g_{i}^{[k]}}{f_i^{[k]}},\\
  \frac{g_{it}^{[k]}}{g_i^{[k]}}=&-(-2{\rm i}\lambda_i^4+\left(2{\rm i}-{\rm i}q_kr_k\right)\lambda_i^2-\frac 18 {\rm i}q_k^2r_k^2+\frac 14 r_kq_{k,x}-\frac 14 q_kr_{k,x}-\frac 12{\rm i})\\
  &-(2q_k\lambda_i^3-(q_k-\frac 12 q_k^2r_k-{\rm i}q_{k,x})\lambda_i)\frac{f_{i}^{[k]}}{g_i^{[k]}},
\end{aligned}
\end{equation}
and
{}{
  \begin{align*}
  \frac{H_{k+1,t}}{H_{k+1}}=&-\frac{1}{4}(4{\rm i}+{\rm i}q_{k+1}^2r_{k+1}^2-2r_{k+1}q_{k+1,x}+2q_{k+1}r_{k+1,x})\\
  =&-4{\rm i}\lambda_i^4+(4{\rm i}-2{\rm i}q_{k}r_{k})\lambda_i^2
  -\frac 14{\rm i}q_{k}^2r_{k}^2-\frac 12(q_kr_{k,x}-r_kq_{k,x})-{\rm i}\\
  &+(2q_k\lambda_i^3-(q_k-\frac 12 q_k^2r_k-{\rm i}q_{k,x})\lambda_i)\frac{f_{i}^{[k]}}{g_i^{[k]}}-(2r_k\lambda_i^3-(r_k-\frac 12 r_k^2q_k+{\rm i}r_{k,x})\lambda_i) \frac{g_{i}^{[k]}}{f_i^{[k]}}.
  \end{align*}}
{}{Above three expressions give }
\begin{equation}
(\frac{g_i^{[k]}}{f_i^{[k]}}H_{k+1})_t=\frac{g_{it}^{[k]}}{g_{i}^{[k]}}-\frac{f_{it}^{[k]}}{f_{i}^{[k]}}
    +\frac{H_{k+1,t}}{H_{k+1}}=0.
\end{equation}
\begin{flushright}
Q.E.D.
\end{flushright}

\end{proof}
Based on  lemma \ref{lemma_H}, let $\dfrac{g_i^{[k]}}{f_i^{[k]}}H_{k+1}=1$ with{}{out} loss of
generality. In this case, {}{$\dfrac{g_1}{f_1}H_{1}=1$, and} then two-fold DT in \eqref{DT2} is simplified as
\begin{equation}
  T_2(\lambda)=T_2(\lambda,\lambda_1,\lambda_2)=\frac{1}{\sqrt{\begin{vmatrix} \lambda_1f_1 &g_1\\
  \lambda_2f_2 &g_2\end{vmatrix}\begin{vmatrix} \lambda_1g_1 &f_1\\ \lambda_2g_2 &f_2 \end{vmatrix}}}
  \left(\begin{matrix} (T_2)_{11} &(T_2)_{12}\\
  (T_2)_{21} &(T_2)_{22}\end{matrix}\right).
\end{equation}
\subsection{The $n$-fold DT} Let us consider the $n$-fold DT for the coupled CLL-NLS
 with the similar method as above. Since
$$T_n(\lambda)=T_n(\lambda,\lambda_1,\lambda_2,\ldots,\lambda_n)=\prod_{k=1}^nT_1^{[k-1]}(\lambda,\lambda_k),$$
let
\begin{equation}
  T_n(\lambda)=T_n(\lambda,\lambda_1,\lambda_2,\ldots,\lambda_n)=\sum_{l=1}^nP_l\sigma_1^{n-l}\lambda^l+P_0,
\end{equation}
where $P_l$ and $\sigma_1$ are defined by
\begin{equation}\nonumber
  P_l=\left(\begin{matrix} P_{l_1} &0\\ 0& P_{l_2}\end{matrix}\right),\quad \sigma_1=\left(\begin{matrix} 0 &1\\ 1&0\end{matrix}\right).
\end{equation}
Furthermore, $P_0$ is determined by
\begin{equation}
  P_0=\prod_{k=1}^n\left(\begin{matrix}0 &-\lambda_k\sqrt{H_{k-1}}\sqrt{\frac{f_{k}^{[k-1]}}{g_{k}^{[k-1]}}}\\
  -\lambda_k\sqrt{H_{k-1}}\sqrt{\frac{g_{k}^{[k-1]}}{f_{k}^{[k-1]}}} &0\end{matrix}\right).
\end{equation}
Here, $H_0=H$, $f_1^{[0]}=f_1$ and $g_1^{[0]}=g_1$. According to lemma \ref{lemma_H}, then
\begin{itemize}\setlength{\itemindent}{-2em}
  \item if $n$ is odd,
  \begin{equation}
    P_0=\left(\begin{matrix}
     &-\lambda_1\lambda_2\lambda_3\ldots\lambda_n\sqrt{H}\sqrt{\frac{f_n^{[n-1]}}{g_n^{[n-1]}}}\\
    -\lambda_1\lambda_2\lambda_3\ldots\lambda_n\frac{1}{\sqrt{H}}\sqrt{\frac{g_n^{[n-1]}}{f_n^{[n-1]}}}&
  \end{matrix}\right),
  \end{equation}\\
  \item if $n$ is even,
  \begin{equation}
    P_0=\left(\begin{matrix}
     \lambda_1\lambda_2\lambda_3\ldots\lambda_n\frac 1 {\sqrt{H}}\sqrt{\frac{f_n^{[n-1]}}{g_n^{[n-1]}}}&\\
   & \lambda_1\lambda_2\lambda_3\ldots\lambda_n\sqrt{{H}}\sqrt{\frac{g_n^{[n-1]}}{f_n^{[n-1]}}}
  \end{matrix}\right).
  \end{equation}
\end{itemize}
%%%%%%%%%%%%%%%%%%%%%%%%%%%%%%%%%%%%%%%%%%%%%%%%%%%%%%%%%%%%%%%%%%%%%%%%%%%%%%%%%%%%%%%%%%%%%%%%%%%%%%%%%%%%%%%%%
%%%%%%%%%%%%%%%%%%%%%%%%%%%%%%%%%%%%%%%%%%%%%%%%%%%%%%%%%%%%%%%%%%%%%%%%%%%%%%%%%%%%%%%%%%%%%%%%%%%%%%%%%%%%%%%%%
\begin{lemma}
  \label{lemma_eigfun}
  After the action of $k$-fold DT, the eigenfunction $\Phi_j$ ($j>k$) related to $\lambda_j$ becomes
  \begin{itemize}\setlength{\itemindent}{-2em}
  \item{if $k$ is odd}
  \begin{equation}
    \Phi_j^{[k]}=\frac 1 {\sqrt{|\Delta_k^1||\Delta_k^2|}}\left(\begin{matrix}\sqrt{H} &\\ &\frac{1}{\sqrt{H}}\end{matrix}\right)
    \left(\begin{matrix} \det(A_{k+1}^1,A_{k+1}^2,A_{k+1}^3,\ldots,A_{k+1}^k,A_{k+1}^j)^T\\
    \det(B_{k+1}^1,B_{k+1}^2,B_{k+1}^3,\ldots,B_{k+1}^k,B_{k+1}^j)^T\end{matrix}\right),
  \end{equation}\\
  \item{if $k$ is even}%%%%%%%%%%%%%%%%%%%%%%%%%%%%%%%%%%%%%%%%%%%%%%%%%%%%%%%%%%%%%%%%%%%%%%%%
  \begin{equation}
    \Phi_j^{[k]}=\frac 1 {\sqrt{|\Delta_k^1||\Delta_k^2|}}
    \left(\begin{matrix} \det(A_{k+1}^1,A_{k+1}^2,A_{k+1}^3,\ldots,A_{k+1}^k,A_{k+1}^j)^T\\
    \det(B_{k+1}^1,B_{k+1}^2,B_{k+1}^3,\ldots,B_{k+1}^k,B_{k+1}^j)^T\end{matrix}\right).
  \end{equation}
  \end{itemize}
\end{lemma}

{\sl Remark:} this lemma is obtained with the inductive method, and the detailed proof is omitted.

Therefore, the explicit expression of $P_0$ is obtained as follows based on lemma \ref{lemma_eigfun}.
\begin{equation}\label{induction}
    P_{0}=
   \begin{cases}
   \left(\begin{matrix}
    \lambda_1\lambda_2\lambda_3\ldots\lambda_k\sqrt{\frac{|\Delta_n^1|}{{|\Delta_n^2|}}}& \\
    &\lambda_1\lambda_2\lambda_3\ldots\lambda_k\sqrt{\frac{|\Delta_n^2|}{|\Delta_n^1|}}
   \end{matrix}\right)&\mbox{if $k$ is even},\\\\
   \left(\begin{matrix}
     \sqrt{H} &\\
     &\frac{1}{\sqrt{H}}
   \end{matrix}\right)\left(\begin{matrix}
     &-\lambda_1\lambda_2\lambda_3\ldots\lambda_k\sqrt{\frac{|\Delta_n^1|}{{|\Delta_n^2|}}} \\
    -\lambda_1\lambda_2\lambda_3\ldots\lambda_k\sqrt{\frac{|\Delta_n^2|}{|\Delta_n^1|}}&
   \end{matrix}\right) &\mbox{if $k$ is odd}.
   \end{cases}
  \end{equation}
{\sl Proof of theorem  \ref{thm_nDT} and \ref{thm_rn}:}
Note that the kernel of $T_n$ consists of  $\Phi_k(k=1,2,\cdots,n)$, i.e., $T_n(\lambda,\lambda_1,\lambda_2,\ldots,\lambda_n)
\Phi_k|_{\lambda=\lambda_k}=0$. Substituting \eqref{induction} into these algebraic
 equations,
the elements of $P_k$ $(k=1,2,\ldots,n)$ in $n$-fold DT are obtained by the Cramer's rule.
This proves theorem \ref{thm_nDT}. Then, theorem \ref{thm_rn} is derived {}{by comparing the coefficient of $\lambda^{n-1}$
in} $T_{nx}+T_nU=U^{[n]}T_n$. \mbox{\hspace{10cm}}
$\Box$

\section{Exact solutions of the CLL-NLS}
In this section, we consider the DT and solution of the coupled CLL-NLS \eqref{cmnls}
under the reduction condition $r=-\overline{q}$, which leads {}{to} the DT and solutions of the CLL-NLS.
\begin{thm}\label{thm_choice}
 Let
\begin{equation}\label{choicesoflambdas}
  \begin{cases}
    \lambda_k=-\overline{\lambda}_k  &\mbox{if $n$ is odd},\\
    \lambda_{2k}=-\overline{\lambda}_{2k-1} &\mbox{if $n$ is even},
  \end{cases}
\end{equation}
then
 solution $(q_n,\,r_n)$ in theorem \ref{thm_rn} preserves the reduction condition $r_n=-\overline{q}_n$.
This means that {}{$T_n$ in theorem \ref{thm_nDT} is a n-fold DT of the CLL-NLS \eqref{MNLS}, and correspondingly,}
$r_n$ is an $n$-th order solution of the CLL-NLS \eqref{MNLS}.
\end{thm}
\begin{proof}
  When $q=-\overline{r}$. From $x$-part of the Lax pair, we have
  %\begin{align}
 % f_x=({\rm i}\lambda^2-\frac 12{\rm i}-\frac 14{\rm i}|r|^2)f+\lambda rg,\quad
  %g_x=-\lambda\overline{r}f-({\rm i}\lambda^2-\frac 12{\rm i}-\frac 14{\rm i}|r|^2)g
  %\end{align}
%and
\begin{equation}
  \overline{f}_x=-({\rm i}\overline{\lambda}^2-\frac 12{\rm i}-\frac 14{\rm i}|r|^2)\overline{f}+\overline{\lambda}\overline{r}\overline{g},\quad
  \overline{g}_x=-\overline{\lambda}r\overline{f}+({\rm i}\overline{\lambda}^2-\frac 12{\rm i}-\frac 14{\rm i}|r|^2)\overline{g}.
\end{equation}
That is
\begin{equation}
  \left(\begin{matrix} \overline{g}_x\\ \overline{f}_x\end{matrix}\right)=\left(\begin{matrix} {\rm i}\overline{\lambda}^2-\frac 12{\rm i}-\frac 14{\rm i}|r|^2 &-\overline{\lambda}r\\ \overline{\lambda}\overline{r} &-{\rm i}\overline{\lambda}^2+\frac 12{\rm i}+\frac 14{\rm i}|r|^2\end{matrix}\right)\left(\begin{matrix} \overline{g}\\ \overline{f}\end{matrix}\right).
\end{equation}
The same property can be obtained from the $t$-part of the Lax pair. Thus, it is obvious
that $\left(\begin{matrix}\overline{g}\\ -\overline{f}\end{matrix}\right)$ is a new
eigenfunction for $\lambda=\overline{\lambda}$ or $\left(\begin{matrix}\overline{g}\\ \overline{f}\end{matrix}\right)$
for $\lambda=-\overline{\lambda}$. For example, $\left(\begin{matrix}\overline{g}_j\\ -\overline{f}_j\end{matrix}\right)$
is a new eigenfunction related to $\lambda_k$ when $\lambda_k=\overline{\lambda}_j$, and
$\left(\begin{matrix}\overline{g}_j\\ \overline{f}_j\end{matrix}\right)$ is another one when $\lambda_k=-\overline{\lambda}_j$.

Based on the above property of the eigenfunctions, we prove that the potentials $q_n$ and $r_n$ will
satisfy the reduction condition, if {}{the choices in \eqref{choicesoflambdas} are adopted in the n-fold DT}.

Note that $\overline{H}=\frac 1H$. For $n=1$, let $\lambda_1=-\overline{\lambda}_1$, then
  \begin{equation*}
  \begin{aligned}
    r_1=&H(\frac {g_1}{f_1}r+2{\rm i}\lambda_1)=H(\frac{\overline{f_1}}{f_1}r+2{\rm i}\lambda_1),\\
    \overline{r}_1=&\frac 1H(\frac{f_1}{\overline{f}_1}\overline{r}-2{\rm i}\overline{\lambda}_1)=-\frac 1H(\frac {f_1}{g_1}q-2{\rm i}\lambda_1)=-q_1.
  \end{aligned}
  \end{equation*}
For $n=2$, let $\lambda_2=-\overline{\lambda}_1$, then $f_2=\overline{g_1}$ and $g_2=\overline{f_1}$. Therefore,
\begin{equation*}
\begin{aligned}
  r_2=&\frac{(\lambda_2g_2f_1-\lambda_1g_1f_2)r+2{\rm i}(\lambda_2^2-\lambda_1^2)f_1f_2}{\lambda_2g_1f_2-\lambda_1f_1g_2}=
 \frac{(-\overline{\lambda}_1{|f_1|}^2-\lambda_1{|g_1|}^2)r+2{\rm i}(\overline{\lambda}_1^2-\lambda_1^2)f_1\overline{g}_1}{-\overline{\lambda}_1|g_1|^2-\lambda_1|f_1|^2},\\
 q_2=&\frac{(\lambda_2f_2g_1-\lambda_1f_1g_2)q-2{\rm i}(\lambda_2^2-\lambda_1^2)g_1g_2}{\lambda_2f_1g_2-\lambda_1g_1f_2}=
 \frac{(-\overline{\lambda}_1|g_1|^2-\lambda_1|f_1|^2)q-2{\rm i}(\overline{\lambda}_1^2-\lambda_1^2)g_1\overline{f}_1}{-\overline{\lambda}_1|f_1|^2-\lambda_1|g_1|^2}=-\overline{r}_2.
\end{aligned}
\end{equation*}
When $n>2$, the reduction condition $q_n=-\overline{r}_n$ can also be obtained by iteration.
\end{proof}

{}{Next,} we provide the solutions of the CLL-NLS, {}{and then discuss their localization characters. In order to achieve this purpose, the eigenfunctions associated with the ``seed" solution depend on the determinant representation of DT. }

\subsection{Eigenfunctions for the Lax pair}
In this subsection, we consider the solution for the Lax pair. Let the seed solution be
\begin{equation}\label{seed}
  q=-\overline{r}=c\exp({\rm i}\rho),\quad \rho=ax+bt,\quad b=a^2+(a-1)c^2, \quad a,c\in\mathbb{R}
\end{equation}
We substitute \eqref{seed} into the Lax pair equations \eqref{Lax} and solving the eigenfunction as follows:
\begin{equation}
\begin{aligned}
  \psi=\left(\begin{matrix}\psi_1(x,t,\lambda)\\ \psi_2(x,t,\lambda)\end{matrix}\right)=\left(\begin{matrix}{D{\rm exp}\left({\rm i}({\frac {\sqrt {S}} 4 x+\frac {\sqrt {S}} 8 \left( -4\,{\lambda}^{2}+2+{c}^{2}+2\,a \right) t-\frac {\rho}2})\right)}\\
 {\frac {D\left( -{\rm i}\sqrt {S}+4\,{\rm i}{\lambda}^{2}-2\,{\rm i}-{\rm i}{c}^{2}+2\,{\rm i}a \right){{\rm exp}\left({\rm i}({\frac {\sqrt {S}} 4 x+\frac {\sqrt {S}} 8 \left( -4\,{\lambda}^{2}+2+{c}^{2}+2\,a \right) t+\frac {\rho}2})\right)} }{4\lambda\,c}}\end{matrix}\right),\\
%%%%%%%%%%%%%%%%%%%%%%%%%%%%%%%%%%%%%%%%%%%%%%%%%%%%%%%%%%%%%%%%%%%%%%%%%%%%%%%%%%%%%%%%%%%%%%%%%%%%%
\end{aligned}
\end{equation}
where $S$ is defined by
\begin{equation*}
  S=16\,{\lambda}^{4}+ \left( 16\,a-16+8\,{c}^{2} \right) {\lambda}^{2}-
4\,{c}^{2}a+4+4\,{c}^{2}-8\,a+{c}^{4}+4\,{a}^{2},
\end{equation*}
and $D$ is a constant. Note that $\left(\begin{matrix} \overline{\psi}_2(x,t,-\overline{\lambda})\\
\overline{\psi}_1(x,t,-\overline{\lambda})\end{matrix}\right)$ is also an eigenfunction under the reduction
condition $r=-\overline{q}$. Thus, we can induce a new eigenfunction by use of the superposition principle:
\begin{equation}\label{eigfunnozero}
  \Phi=\left(\begin{matrix} f(x,t,\lambda)\\g(x,t,\lambda)\end{matrix}\right)=\left(\begin{matrix} \psi_1(x,t,\lambda)+\overline{\psi}_2(x,t,-\overline{\lambda})\\
  \psi_2(x,t,\lambda)+\overline{\psi}_1(x,t,-\overline{\lambda})
  \end{matrix}\right).
\end{equation}
Let $\lambda=\lambda_j$, then $\Phi$ \eqref{eigfunnozero} leads to the eigenfunction $\Phi_j$ related to $\lambda_j$.
Furthermore, when $q=-\overline{r}=c\exp({\rm i}\rho)$, the explicit expression of $H$ is given in \eqref{HHE}.

Now, with the help of theorems \ref{thm_rn} and \ref{thm_choice}, we can present the solutions of the CLL-NLS.

\subsection{Soliton, breather and first-order rogue wave solutions}For $n=1$, let $\lambda_1={\rm i}\beta$ and $D=1$ in
{}{theorem \ref{thm_rn}},
then the first-order solution of the CLL-NLS is
  \begin{equation}\label{dbsoliton}
    r_1=(\frac{L_1\cos\theta+{\rm i}L_2\sin\theta}{-L_1\cos\theta+{\rm i}L_2\sin\theta}c-2\beta)H,
  \end{equation}
with

    \begin{align*}
    \theta=&(\frac 12\,{\beta}^{2}+\frac 14+\frac 18\,{c}^{2}+\frac 14\,a)\sqrt {S_{{1}}}\,t+\frac 14\sqrt {S_{{1}}}\,x,\\
      L_1=&-4\,\beta\,c+\sqrt {{S_1}}+4\,{\beta}^{2}+2+{c}^{2}-2\,a,\quad
      L_2=4\,\beta\,c+\sqrt {{S_1}}+4\,{\beta}^{2}+2+{c}^{2}-2\,a,\\
      S_1=&16\,{\beta}^{4}+ \left( -8\,{c}^{2}-16\,a+16 \right) {\beta}^{2}-4\,{c
}^{2}a+4+4\,{c}^{2}-8\,a+{c}^{4}+4\,{a}^{2}.
    \end{align*}

It is obvious that $r_1$ leads to a periodic solution, if $S_1>0$ and gives a soliton solution if $S_1<0$.
When $x$ and $t$ tend both to infinity, $|r_1|^2$ tends to $2a-2$ $(\mbox{here } a>1)$.  When $S_1<0$, $|r_1|$
reaches to its amplitude  of $|2\beta+c|$
at $x=-\frac 12 \left( 4\,{\beta}^{2}+2+{c}^{2}+2\,a \right) t$. Thus, if $2a-2>|2\beta+c|^2$, it gives
a dark soliton. Otherwise, it leads {}{to} a bright solitonic localization with a non-vanishing boundary. That is, the CLL-NLS can
give both bright soliton and dark soliton. This is different from the NLS, that depends on the signs of the dispersion and nonlinear parameter in order to admit rather dark or bright soliton solutions. The bright soliton and dark soliton solutions are shown in Fig. \ref{fig_soliton}.

\begin{figure}[!htbp]
\centering
\raisebox{20 ex}{$|r_1|$}\subfigure[]{\includegraphics[height=6cm,width=6cm]{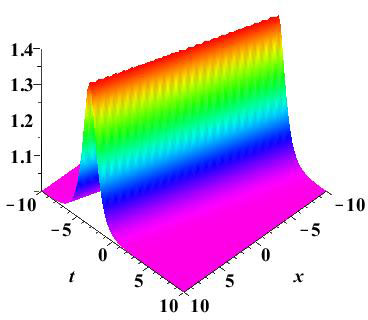}}
\qquad
\raisebox{20 ex}{$|r_1|$}\subfigure[]{\includegraphics[height=6cm,width=6cm]{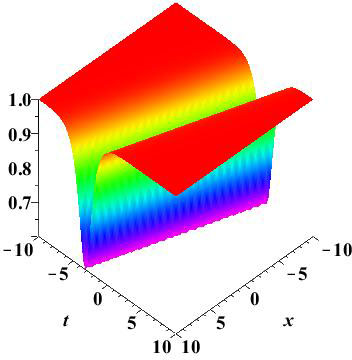}}\caption{(Color online) The bright and dark solitons of the CLL-NLS with parameters: (a) $a=1.5,c=1$ and $\beta=0.2$, (b) $a=1.5,c=1$ and $\beta=-0.2$.}\label{fig_soliton}
\end{figure}

For $n=2$, let $D=1$, $\lambda_1=\alpha_1+{\rm i}\beta_1$ and $\lambda_2=-\alpha_1+{\rm i}\beta_1$ in {}{theorem \ref{thm_rn}}, then
\begin{equation}
  r_2=\frac{(-\overline{\lambda}_1{|f_1|}^2-\lambda_1{|g_1|}^2)r+2{\rm i}(\overline{\lambda}_1^2-\lambda_1^2)f_1\overline{g}_1}{-\overline{\lambda}_1|g_1|^2-\lambda_1|f_1|^2}
\end{equation}
gives the second-order solution of the CLL-NLS. For {}{convenience}, let $a=2\,{\beta_{{1}}}^{2}-\frac 12\,{c}^{2}-2\,{\alpha_{{1}}}^{2}+1$, then
\begin{equation}\label{breather}
  r_2=-\frac{K_1c\cos\theta_1+{\rm i}K_2c\sin\theta_1+(K_3c+K_5)\cosh\theta_2+{\rm i}(K_6-K_4c)\sinh\theta_2}{-K_1\cos\theta_1+{\rm i}K_2\sin\theta_1+K_3\cosh\theta_2+{\rm i}K_4\sinh\theta_2}\exp(-{\rm i}\rho)
\end{equation}
with

  \begin{align*}
  \theta_1=&((4\alpha_1^2-4\beta_1^2-2)t-x)K_0,\quad%\sqrt{(c^2+4\alpha_1^2)(c^2-4\beta_1^2)},\\
  \theta_2=4\alpha_1\beta_1tK_0,\quad%\sqrt{(c^2+4\alpha_1^2)(c^2-4\beta_1^2)},\\
    K_0=\sqrt{(c^2+4\alpha_1^2)(c^2-4\beta_1^2)},\\
    K_1=&8\alpha_1^3\beta_1+2c^2\alpha_1\beta_1+2\alpha_1\beta_1K_0,\quad%\sqrt{(c^2+4\alpha_1^2)(c^2-4\beta_1^2)},\\
    K_2=2c^2\alpha_1\beta_1-8\alpha_1\beta_1^3+2\alpha_1\beta_1K_0,\\%\sqrt{(c^2+4\alpha_1^2)(c^2-4\beta_1^2)},\\
    K_3=&c^3\alpha_1+4c\alpha_1^3+c\alpha_1K_0,\quad%\sqrt{(c^2+4\alpha_1^2)(c^2-4\beta_1^2)},\\
    K_4=4c\beta_1^3-c^3\beta_1-c\beta_1K_0,\\%\sqrt{(c^2+4\alpha_1^2)(c^2-4\beta_1^2)},\\
    K_5=&-8c^2\alpha_1\beta_1^2-32\alpha_1^3\beta_1^2-8\alpha_1\beta_1^2K_0,\quad%\sqrt{(c^2+4\alpha_1^2)(c^2-4\beta_1^2)},\\
    K_6=8c^2\alpha_1^2\beta_1-32\alpha_1^2\beta_1^3+8\alpha_1^2\beta_1K_0.%\sqrt{(c^2+4\alpha_1^2)(c^2-4\beta_1^2)}.
  \end{align*}

Note that the trajectory of $r_2$ is defined by
$$(4\alpha_1^2-4\beta_1^2-2)t-x=0,$$
if $K_0^2<0$, and by
$$4\alpha_1\beta_1t=0,$$
if $K_0^2>0$. Thus, we can get both the spatial periodic breather solution (similar to the NLS Akhmediev breather \cite{Abreather})
and the temporal periodic breather solution (similar to the NLS Kuznetsov-Ma breather \cite{Kuznetsovbreather,Mabreather}). In fact, this solution $r_2$ can
travel periodically with an additional velocity in
the $(x,t)$-plane. Three kinds of breather solutions, propagating along the $(x,t)$-plane with different angles, are shown
in Fig. \ref{fig_breather}.

\begin{figure}[!htbp]
\centering
\raisebox{20 ex}{$|r_2|$}\subfigure[]{\includegraphics[height=6cm,width=6cm]{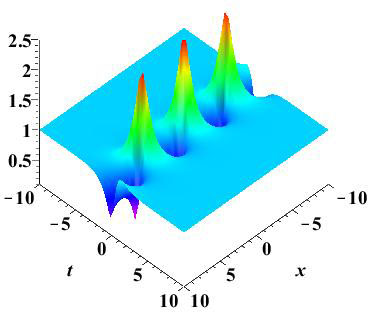}}
\qquad
\raisebox{20 ex}{$|r_2|$}\subfigure[]{\includegraphics[height=6cm,width=6cm]{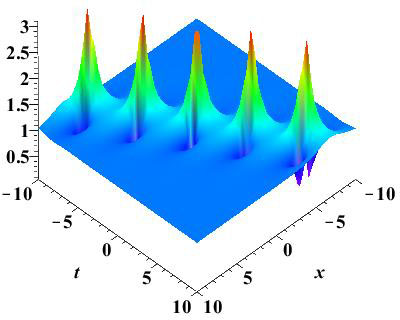}}
\qquad
\raisebox{20 ex}{$|r_2|$}\subfigure[]{\includegraphics[height=6cm,width=6cm]{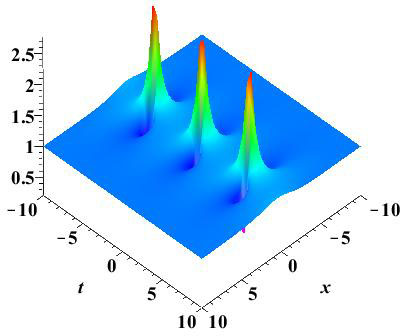}}
\caption{(Color online) The breather solutions of the CLL-NLS with parameters: (a) $c=1$, $\alpha=0.75$ and $\beta=0.4$; (b) $c=1$, $\alpha=0.8$ and $\beta=0.55$; (c) $c=1$, $\alpha_1^2=\beta_1^2+\frac 12$, and $\beta_1=0.52$.}\label{fig_breather}
\end{figure}

After a simple analysis, we emphasize that the periodicity of the breather solution
is proportional to $\frac{1}{K_0}$, i.e. when $K_0$ tends to zero,
 the distance of two peaks goes to
infinity which leaves only one peak locating on the $(x,t)$-plane. Thus, let $c\rightarrow2\beta_1$,
then $r_2$ in \eqref{breather} leads {}{to} a new solution, having the property to possess only one local peak and surrounding two holes
which is very similar to the Peregrine solution and therefore, being an appropriate to model RWs. This kind of doubly-localized rational solution is described by
\begin{equation}\label{1rw}
  r_{2r}=\frac{2\beta_1(L_{11}+{\rm i}L_{12})}{L_{11}+{\rm i}L_{13}+4}\exp(-{\rm i}\rho),
\end{equation}
with

\begin{align*}
  L_{11}=&(16\alpha_1^2\beta_1^2+16\beta_1^4)x^2-(128\alpha_1^4\beta_1^2-64\beta_1^4-64\alpha_1^2\beta_1^2-128\beta_1^6)xt
  -(256\alpha_1^4\beta_1^2-64\alpha_1^2\beta_1^2\\
  &-256\beta_1^8-256\alpha_1^6\beta_1^2-256\beta_1^6-64\beta_1^4)t^2-3,\\
  L_{12}=&8\beta_1^2x+16\beta_1^2t-96\alpha_1^2\beta_1^2t,\quad
  L_{13}=(32\alpha_1^2\beta_1^2-64\beta_1^4-16\beta_1^2)t-8\beta_1^2x.
\end{align*}

When $x$ and $t$ tend to infinity, $|r_{2r}|$  tends to $2\beta_1$. Moreover, the maximum peak amplitude
 is equal to $6\beta_1$, which is three times the background amplitude. The profiles are shown in Fig. \ref{fig_1rw}, and this solution is the same as presented in \cite{Chow1}. The latter has been derived using the Hirota bilinear method, while difficulties using the DT have been also discussed in \cite{Chow1}.

\subsection{Higher-order rogue wave solutions}
Inspired by above method, we consider \sout{the} higher-order
RWs of the CLL-NLS in this subsection. Generally, it is difficult to derive higher-order RWs from multi-breather solutions, since the explicit expression of $n$-th order breather is very challenging to calculate when $n>4$. Similarly for the NLS, for which the formulae is given by theorem \ref{thm_rn}, an indeterminate form $\frac{0}{0}$ is a consequence, when eigenvalues $\lambda_k$ tend to a limit point (from a breather to a doubly-localized RW solution). Thus, we derive the higher-order RWs directly from the determinant expressions of solutions in theorem \ref{thm_rn} by adopting a Taylor expansion \cite{He1,Guoo1,Guo2,Zhang1}.
\begin{thm}\label{thm_rw}
  Let $n=2N$, $\lambda_{2k-1}=\sqrt{\frac{1-a}{2}}+\frac{{\rm i}c}2+\epsilon^2$ $(a<1)$ and $\lambda_{2k}=-\overline{\lambda}_{2k-1}$, by applying the Taylor expansion, then {}{a} determinant expression of the $N$-th order RW is given as
  \begin{equation}
    r_{nr}=\frac{|\widehat{\Delta}_n^1|}{|\widehat{\Delta}_n^2|}r-2{\rm i}\frac{|\widehat{\Delta}_n^4|}{|\widehat{\Delta}_n^2|},
  \end{equation}
where $\widehat{\Delta}_n^k$ $(k=1,2,3,4)$ are defined by
\begin{equation}
  \widehat{\Delta}_n^k=\left(\left.\frac{\partial^{n_i}}{\partial\epsilon^{n_i}}\right|_{\epsilon=0}
      ({\Delta}_n^k)_{ij}\right)_{n\times n}.
\end{equation}
Here, $n_i=i$ if $i$ is odd and $n_i=i-1$ if $i$ is even.
\end{thm}
Note that $D$ is a constant in \eqref{eigfunnozero}, it is reasonable to choose $D=\exp({\rm i}\sqrt{S}
(\sum^{N-1}_{l=1}s_l\epsilon^{2l})$. Here, $D$ goes to $1$ when $\epsilon$ goes to zero.
Thus, there exist $N+1$ free parameters {}{($s_1,s_2,\cdots, s_{N-1}; a,c$)} in an  $N$-th order RW solution. Next, we derive
RWs with these parameters, and consider their dynamical evolution. For convenience, let $a=-1$
and $c=1$ in the following context.

For $N=2$, the second-order RW of the CLL-NLS is
\begin{equation}
  r_{4r}=\frac{L_{21}}{L_{22}}\exp(-{\rm i}\rho),
\end{equation}
where

\begin{align*}
  L_{21}=&125\,{x}^{6}+150\,{\rm i}{x}^{5}-750\,t{x}^{5}-285\,{x}^{4}+3375\,{t}^{2}{x}
^{4}-2100\,{\rm i}{x}^{4}t-156\,{\rm i}{x}^{3}-8500\,{t}^{3}{x}^{3}\\
&+2220\,t{x}^{3}+8100\,{\rm i}{t}^{2}{x}^{3}-24000\,{\rm i}{t}^{3}{x}^{2}-14850\,{t}^{2}{x}^{2}-
333\,{x}^{2}+16875\,{t}^{4}{x}^{2}+2160\,{\rm i}t{x}^{2}\\
&-270\,tx-18750\,{t}^{5}x+33750\,{\rm i}{t}^{4}x+28500\,{t}^{3}x-8100\,{\rm i}x{t}^{2}-90\,{\rm i}x+45+15625
\,{t}^{6}\\
&+1044\,{\rm i}t+600\,{\rm i}{t}^{3}-2205\,{t}^{2}-37500\,{\rm i}{t}^{5}-26625\,
{t}^{4}+(-300\,{x}^{3}-1800\,{\rm i}{x}^{2}+4500\,t{x}^{2}\\
&+1800\,{\rm i}tx+180\,x-4500\,{t}
^{2}x-4500\,{t}^{3}-432\,{\rm i}+540\,t+7200\,{\rm i}{t}^{2})s_1,\\
L_{22}=&-125\,{x}^{6}+750\,t{x}^{5}+150\,{\rm i}{x}^{5}-15\,{x}^{4}-3375\,{t}^{2}{x}
^{4}-600\,{\rm i}{x}^{4}t+180\,t{x}^{3}+8500\,{t}^{3}{x}^{3}\\
&+84\,{\rm i}{x}^{3}+2100\,{\rm i}{t}^{2}{x}^{3}+2250\,{t}^{2}{x}^{2}-3000\,{\rm i}{t}^{3}{x}^{2}-16875
\,{t}^{4}{x}^{2}-171\,{x}^{2}-360\,{\rm i}{x}^{2}t\\
&-4500\,{t}^{3}x+18750\,{t}^{5}x-900\,{\rm i}x{t}^{2}+3750\,{\rm i}{t}^{4}x+270\,tx+54\,{\rm i}x-15625\,{t}^{6}+144
\,{\rm i}t\\
&+3600\,{\rm i}{t}^{3}-3195\,{t}^{2}-10875\,{t}^{4}-9
+(300\,{x}^{3}-4500\,t{x}^{2}+4500\,{t}^{2}x-180\,x+1800\,{\rm i}tx\\
&+4500\,{t}^{3}-1800\,{\rm i}{t}^{2}-72\,{\rm i}+3060\,t)s_1,
\end{align*}
where $s_1$ is a free complex parameter. The maximum amplitude of $|r_{4r}|$ is equal to $5$ when $s_1=0$, which
is reached at $(x=0,\,t=0)$. This solution is shown in Fig. \ref{fig_2rw}(a).  Allocating different values to
$s_1$, we can obtain RWs which {}{are} distinct from the above one. For example, RWs with $s_1=100-100{\rm i}$
and $s_1=100+100{\rm i}$ are shown in Fig. \ref{fig_2rw}(b) and Fig. \ref{fig_2rw}(c), respectively. Both of them
possess three intensity peaks, located at different time and space values. Each peak is similar to a first-order RW, shown in Fig. \ref{fig_1rw}(a). Moreover, solution $r_{4r}$ in Fig. \ref{fig_2rw}(b) is different from the one in Fig. \ref{fig_2rw}(c), since  three peaks in each solution are arrayed in different directions.

For $N=3$, according to theorem \ref{thm_rw}, we can obtain the third-order RW solution of the CLL-NLS
equation. However, its expression, with two {}{non-zero} parameters $s_1$ and $s_2$, is very cumbersome, that we just provide the exact expression
in the case $s_1=s_2=0$, i.e.
\begin{equation}
  r_{6r}=\frac{L_{31}}{L_{32}}\exp(-{\rm i}\rho),
\end{equation}
 with
  \begin{align*}
    L_{31}=&-3125\,{x}^{12}+37500\,t{x}^{11}-7500\,{\rm i}{x}^{11}-281250\,{t}^{2}{x}^{
           10}+150000\,{\rm i}t{x}^{10}+18750\,{x}^{10}+1437500\,{t}^{3}{x}^{9}\\
           &-1237500\,{\rm i}{t}^{2}{x}^{9}+22500\,{\rm i}{x}^{9}-322500\,t{x}^{9}+3386250\,{t}^{2}{x}^{8}+6975000\,{\rm i}{t}^{3}{x}^{8}-5671875\,{t}^{4}{x}^{8}\\
           &-495000\,{\rm i}t{x}^{8}+31725\,{x}^{8}+5130000\,{\rm i}{t}^{2}{x}^{7}+63720\,{\rm i}{x}^{7}+17475000\,{t}^{5}{x}^{7}+84600\,t{x}^{7}\\
           &-19890000\,{t}^{3}{x}^{7}-27975000\,{\rm i}{t}^{4}{x}^{7}+87240000\,{\rm i}{t}^{5}{x}^{6}+82807500\,{t}^{4}{x}^{6}+116676\,{x}^{6}-858240\,{\rm i}t{x}^{6}\\
           &-33000000\,{\rm i}{t}^{3}{x}^{6}-3338100\,{t}^{2}{x}^{6}-43637500\,{t}^{6}{x}^{6}+2230200\,{\rm i}{t}^{2}{x}^{5}+87375000\,{t}^{7}{x}^{5}\\
           &+26829000\,{t}^{3}{x}^{5}+31320\,t{x}^{5}+133155000\,{\rm i}{t}^{4}{x}^{5}-209775000\,{\rm i}{t}^{6}{x}^{5}-245835000\,{t}^{5}{x}^{5}+129384\,{\rm i}{x}^{5}\\
           &+401250000\,{\rm i}{t}^{7}{x}^{4}+550312500\,{t}^{6}{x}^{4}-146475\,{x}^{4}-1717200\,{\rm i}t{x}^{4}-366750000\,{\rm i}{t}^{5}{x}^{4}\\
           &-93881250\,{t}^{4}{x}^{4}-607500\,{t}^{2}{x}^{4}-141796875\,{t}^{8}{x}^{4}-3330000\,{\rm i}{t}^{3}{x}^{4}+1315980\,t{x}^{3}\\
           &-52380\,{\rm i}{x}^{3}-585937500\,{\rm i}{t}^{8}{x}^{3}+696450000\,{\rm i}{t}^{6}{x}^{3}+4222800\,{\rm i}{t}^{2}{x}^{3}+6975000\,{\rm i}{t}^{4}{x}^{3}\\
           &+179687500\,{t}^{9}{x}^{3}-902250000\,{t}^{7}{x}^{3}+209925000\,{t}^{5}{x}^{3}+9702000\,{t}^{3}{x}^{3}
           -846000000\,{\rm i}{t}^{7}{x}^{2}\\
           &-12429450\,{t}^{2}{x}^{2}-175781250\,{t}^{10}{x}^{2}+656250000\,{\rm i}{t}^{9}{x}^{2}+1078593750\,{t}^{8}{x}^{2}-19872000\,{\rm i}{t}^{3}{x}^{2}\\
           &+126900000\,{\rm i}{t}^{5}{x}^{2}-332212500\,{t}^{6}{x}^{2}+17347500\,{t}^{4}
           {x}^{2}-62370\,{x}^{2}+887760\,{\rm i}t{x}^{2}\\
           &+117187500\,{t}^{11}x-11340\,{\rm i}x+25825500\,{t}^{3}x-299475000\,{\rm i}{t}^{6}x-286740\,tx+52245000\,{\rm i}{t}^{4}x\\
           &+621562500\,{\rm i}{t}^{8}x-5046300\,{\rm i}x{t}^{2}-492187500\,{\rm i}{t}^{10}x+313875000\,{t}^{7}x-128925000\,{t}^{5}x\\
           &-839062500\,{t}^{9}x-190350000\,{\rm i}{t}^{5}-46875000\,{\rm i}{t}^{9}-715230\,{t}^{2}-48828125\,{t}^{12}+2835\\
           &-5761800\,{\rm i}{t}^{3}+222328125\,{t}^{8}-47833875\,{t}^{4}+53550000\,{\rm i}{t}^{7}+234375000\,{\rm i}{t}^{11}+363281250\,{t}^{10}\\
           &+291937500\,{t}^{6}+158760\,{\rm i}t,\\
%%%%%%%%%%%%%%%%%%%%%%%%%%%%%%%%%%%%%%%%%%%%%%%%%%%%%%%%%%%%%%%%%%%%%%%%%%%%%%%%%%%%%%%%%%%%%%%%%%%%%%%%%%%%%%
L_{32}=&-3125\,{x}^{12}+37500\,t{x}^{11}+7500\,{\rm i}{x}^{11}-75000\,{\rm i}t{x}^{10}+3750\,{x}^{10}-281250\,{t}^{2}{x}^{10}+487500\,{\rm i}{t}^{2}{x}^{9}\\
&+7500\,{\rm i}{x}^{9}+1437500\,{t}^{3}{x}^{9}-22500\,t{x}^{9}-5671875\,{t}^{4}{x}^{8
}-13275\,{x}^{8}-90000\,{\rm i}t{x}^{8}+461250\,{t}^{2}{x}^{8}\\
&-2100000\,{\rm i}{t}
^{3}{x}^{8}-3090000\,{t}^{3}{x}^{7}+6975000\,{\rm i}{t}^{4}{x}^{7}+84600\,t{
x}^{7}+37080\,{\rm i}{x}^{7}+17475000\,{t}^{5}{x}^{7}\\
&-90000\,{\rm i}{t}^{2}{x}^{7}
-998100\,{t}^{2}{x}^{6}-43637500\,{t}^{6}{x}^{6}+13057500\,{t}^{4}{x}^
{6}-17490000\,{\rm
 i}{t}^{5}{x}^{6}-239760\,{\rm i}t{x}^{6}\\
&+2100000\,{\rm i}{t}^{3}{x}^{
6}-81324\,{x}^{6}+34875000\,{\rm i}{t}^{6}{x}^{5}+4509000\,{t}^{3}{x}^{5}+
97848\,{\rm i}{x}^{5}-35955000\,{t}^{5}{x}^{5}\\
&+489240\,t{x}^{5}+1290600\,{\rm i}{t
}^{2}{x}^{5}+87375000\,{t}^{7}{x}^{5}-10215000\,{\rm i}{t}^{4}{x}^{5}-
3469500\,{t}^{2}{x}^{4}+11205\,{x}^{4}\\
&+28800000\,{\rm i}{t}^{5}{x}^{4}-
7031250\,{t}^{4}{x}^{4}-52500000\,{\rm i}{t}^{7}{x}^{4}+62062500\,{t}^{6}{x}
^{4}-141796875\,{t}^{8}{x}^{4}\\
&-3780000\,{\rm i}{t}^{3}{x}^{4}-216000\,{\rm i}t{x}^
{4}-48450000\,{\rm i}{t}^{6}{x}^{3}+961200\,{\rm i}{t}^{2}{x}^{3}+106380\,t{x}^{3}
+26460\,{\rm i}{x}^{3}\\
&+60937500\,{\rm i}{t}^{8}{x}^{3}-3375000\,{\rm i}{t}^{4}{x}^{3}+
1125000\,{t}^{5}{x}^{3}-62250000\,{t}^{7}{x}^{3}+9702000\,{t}^{3}{x}^{
3}+179687500\,{t}^{9}{x}^{3}\\
&+5724000\,{\rm i}{t}^{3}{x}^{2}-73102500\,{t}^{4
}{x}^{2}-87480\,{\rm i}t{x}^{2}-46875000\,{\rm i}{t}^{9}{x}^{2}-158512500\,{t}^{6}
{x}^{2}+67500000\,{\rm
 i}{t}^{7}{x}^{2}\\
&-20250\,{x}^{2}+1891350\,{t}^{2}{x}^{
2}-175781250\,{t}^{10}{x}^{2}-18281250\,{t}^{8}{x}^{2}+66150000\,{\rm i}{t}^
{5}{x}^{2}+4860\,{\rm i}x\\
&+11340\,tx+23437500\,{\rm i}{t}^{10}x+98437500\,{t}^{9}x-
45225000\,{\rm i}{t}^{6}x-882900\,{\rm i}{t}^{2}x-526500\,{t}^{3}x\\
&+146475000\,{t}^
{5}x+13095000\,{\rm i}{t}^{4}x+117187500\,{t}^{11}x-47812500\,{\rm i}{t}^{8}x+
331875000\,{t}^{7}x\\
&-345796875\,{t}^{8}-109912500\,{t}^{6}-405-8100000
\,{\rm i}{t}^{5}+32400\,{\rm i}t-48828125\,{t}^{12}-990630\,{t}^{2}\\
&-24883875\,{t}^
{4}+56250000\,{\rm i}{t}^{9}-152343750\,{t}^{10}+52200000\,{\rm i}{t}^{7}+4276800
\,{\rm i}{t}^{3}.
  \end{align*}

\begin{figure}[!htbp]
\centering
\raisebox{20 ex}{$|r_{2r}|$}\subfigure[]{\includegraphics[height=6cm,width=6cm]{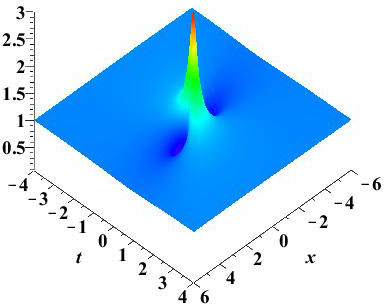}}
\qquad
\raisebox{20 ex}{$|r_{2r}|$}\subfigure[]{\includegraphics[height=6cm,width=6cm]{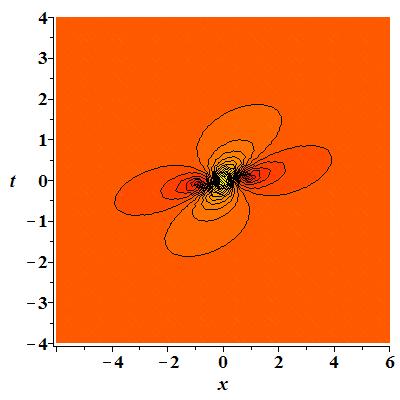}}
\caption{(Color online) The first-order RW solution of the CLL-NLS with parameters: $\alpha_1=-1,\,\beta_1=0.5$. The right panel is the density plot of the left.}\label{fig_1rw}
\end{figure}

\begin{figure}[!htbp]
\centering
\raisebox{20 ex}{$|r_{4r}|$}\subfigure[]{\includegraphics[height=6cm,width=6cm]{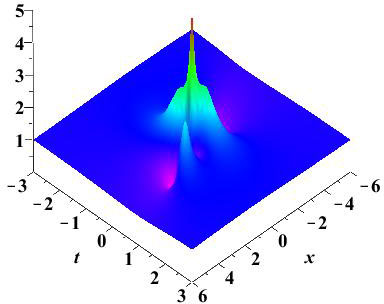}}
\qquad
\raisebox{20 ex}{$|r_{4r}|$}\subfigure[]{\includegraphics[height=6cm,width=6cm]{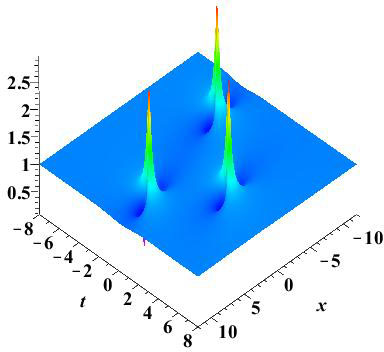}}
\qquad
\raisebox{20 ex}{$|r_{4r}|$}\subfigure[]{\includegraphics[height=6cm,width=6cm]{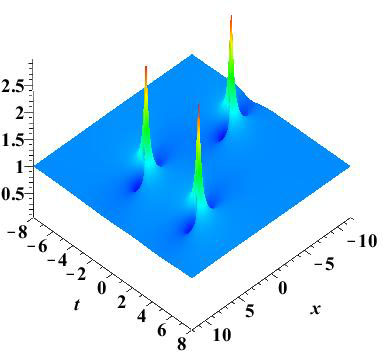}}
\caption{(Color online) The second-order RW of the CLL-NLS with parameters: (a) $s_1=0$; (b) $s_1=100-100{\rm i}$; (c) $s_1=100+100{\rm i}$.}\label{fig_2rw}
\end{figure}

The maximum amplitude is equal to $7$ which occurs at $(0,\,0)$, its profile is shown in Fig. \ref{fig_3rw}(a). Let $s_1\neq0$ and $s_2\neq0$, we obtain {}{other} solutions which are different from the one given in Fig. \ref{fig_3rw}(a). In each of these solutions, the third-order RW is split into six intensity peaks which are similar to a first-order RW. These six peaks, located at different point of time and space, make up different profiles. As example, three such solutions are displayed in Fig. \ref{fig_3rw}(b-d) with $(s_1,\, s_2)=(100,\,0)$,  $(s_1,\,s_2)=(0,\,5000)$, and $(s_1,\,s_2)=(100,\,13000)$, respectively. In Fig. \ref{fig_3rw}(b), these six intensity form a triangle. In Fig. \ref{fig_3rw}(c), they compose a pentagon with five peaks locating on the shell and the other one locating on
 the center. In Fig. \ref{fig_3rw}(d), three peaks compose a triangle and the other three peaks compose a  part of a circular arc.

\begin{figure}[!htbp]
\centering
\raisebox{20 ex}{$|r_{6r}|$}\subfigure[]{\includegraphics[height=6cm,width=6cm]{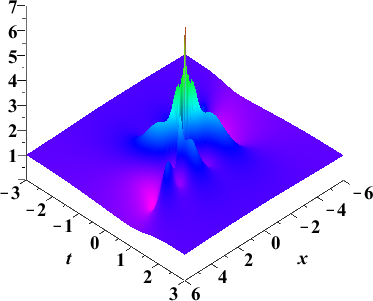}}
\qquad
\raisebox{20 ex}{$|r_{6r}|$}\subfigure[]{\includegraphics[height=6cm,width=6cm]{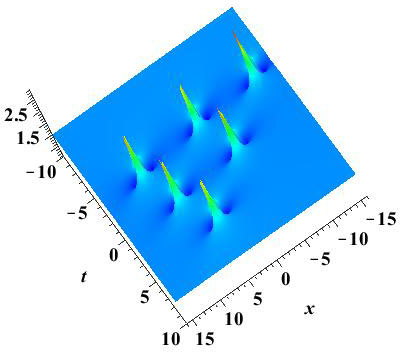}}
\\
\raisebox{20 ex}{$|r_{6r}|$}\subfigure[]{\includegraphics[height=6cm,width=6cm]{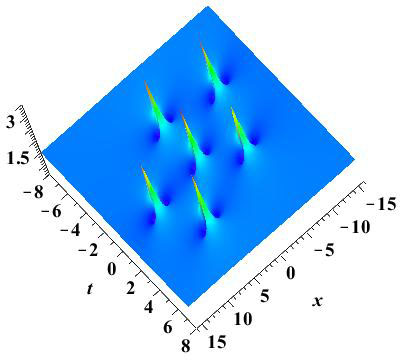}}
\qquad
\raisebox{20 ex}{$|r_{6r}|$}\subfigure[]{\includegraphics[height=6cm,width=6cm]{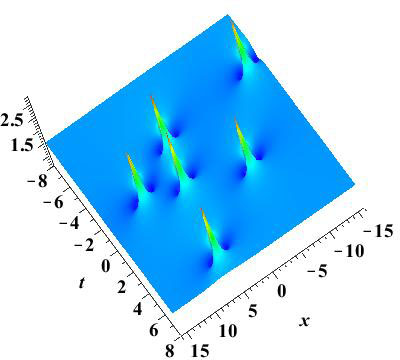}}
\caption{(Color online) The third-order RW of the CLL-NLS with parameters $(s_1,\,s_2)$ as: (a) $(0,0)$; (b) $(100,0)$; (c) $(0,5000)$; (d) $(100,13000)$.}\label{fig_3rw}
\end{figure}

Let $N=4$ in theorem \ref{thm_rw}, then $r_{8r}$ gives {}{a} fourth-order RW of the CLL-NLS with three
parameters $s_1,\,s_2$ and $s_3$. Let $(s_1,\,s_2,\,s_3)=(0,\,0,\,0)$, $r_{8r}$ leads {}{to} a solution {}{with a  highest peak
surrounded by several gradually decreasing peaks in two sides along $t$-direction, which is the fundamental pattern} and is shown in Fig \ref{fig_4rw}(a). The amplitude of this solution is $9$ located at the origin of coordinate. Furthermore, allocating different values to $(s_1,\,s_2,\,s_3)$, we obtain  a hierarchy of solutions, which have a triangle pattern, a pentagon pattern, a circular pattern with a inner second-order fundamental pattern or triangle pattern.
 These solutions are shown in Fig. \ref{fig_4rw}(b-c) and Fig. \ref{fig_4rw_1}.

\begin{figure}[!htbp]
\centering
\raisebox{20 ex}{$|r_{8r}|$}\subfigure[]{\includegraphics[height=6cm,width=6cm]{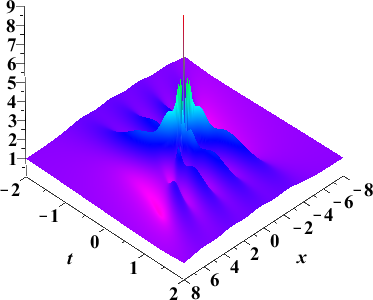}}
\qquad
\raisebox{20 ex}{$|r_{8r}|$}\subfigure[]{\includegraphics[height=6cm,width=6cm]{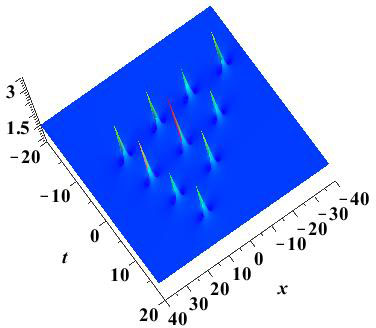}}
\qquad
\raisebox{20 ex}{$|r_{8r}|$}\subfigure[]{\includegraphics[height=6cm,width=6cm]{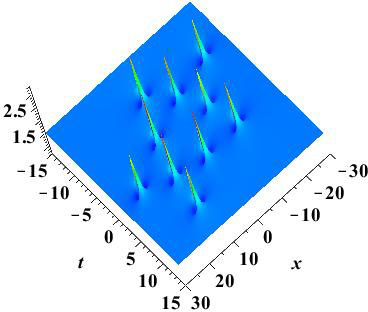}}
\caption{(Color online) The fourth-order RW of the CLL-NLS with parameters $(s_1,\,s_2,\,s_3)$ as: (a) $(0,0,0)$; (b) $(500,0,0)$; (c) $(0,50000,0)$.}\label{fig_4rw}
\end{figure}

\begin{figure}[!htbp]
\centering
\raisebox{20 ex}{$|r_{8r}|$}\subfigure[]{\includegraphics[height=6cm,width=6cm]{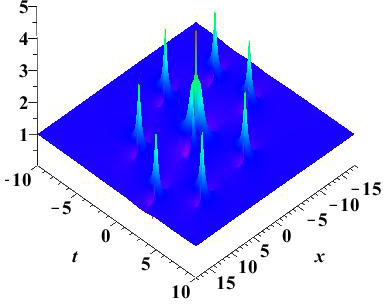}}
\qquad
\raisebox{20 ex}{$|r_{8r}|$}\subfigure[]{\includegraphics[height=6cm,width=6cm]{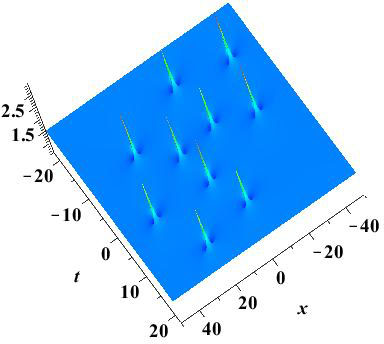}}
\caption{(Color online) The fourth-order RW in circular pattern of the CLL-NLS with parameters $(s_1,\,s_2,\,s_3)$ as: (a) $(0,0,500000)$; (b) $(500,0,50000000)$.}\label{fig_4rw_1}
\end{figure}

All the results are derived as a consequence of theorem \ref{thm_rw} and can be trivially extended to the higher-order RWs. That is, the explicit
expressions of other higher-order solutions can be obtained in a straightforward manner. However, we will omit this, since expressions are too
cumbersome to be explicitly written here.  All solutions, presented above, have been verified analytically by symbolic computation through a Maple computer software.

\section{Localization characters of CLL-NLS rogue waves}
In this section, we consider the localization characters of the RW of the CLL-NLS as well as the {influence} of SSE on these localization. First, we need to define the length and width of
the RW solution as described in \cite{PRE1}. In order to compare the latter properties with localization of NLS RWs \cite{He1,PRE1}, we replace the parameters $\alpha_1$
 and $\beta_1$ with $a$ and $c$ in \eqref{1rw}. That is, we substitute $\alpha_1=\sqrt{\frac{1-a}{2}}$ $(a<1)$ and
 $\beta_1=\frac c2$ into \eqref{1rw}. In this case, the first-order RW {}{of the CLL-NLS} is expressed as {}{the following}
\begin{equation}
  r_{2r}=\frac{L_n}{L_d}c\exp(-{\rm i}\rho),
\end{equation}
with
\begin{align*}
  L_n=&3-{c}^{4}{x}^{2}-4\,{t}^{2}{c}^{4}-4\,{c}^{6}{t}^{2}-{c}^{8}{t}^{2}-2
\,{c}^{6}xt+8\,{\rm i}t{c}^{2}-12\,{\rm i}{c}^{2}ta-4\,{c}^{4}xt+8\,{c}^{2}{t}^{2}
{a}^{3}-2\,{c}^{2}{x}^{2}\\
      &+2\,a{c}^{2}{x}^{2}-8\,{c}^{2}xta-2\,{\rm i}{c}^{2}
x-8\,{c}^{2}{t}^{2}{a}^{2}+8\,{c}^{2}xt{a}^{2},\\
  L_d=&-1-{c}^{4}{x}^{2}-4\,{t}^{2}{c}^{4}-4\,{c}^{6}{t}^{2}-{c}^{8}{t}^{2}-2
\,{c}^{6}xt+4\,{\rm i}{c}^{2}ta+4\,{\rm i}{c}^{4}t-4\,{c}^{4}xt+8\,{c}^{2}{t}^{2}{
a}^{3}-2\,{c}^{2}{x}^{2}\\
      &+2\,a{c}^{2}{x}^{2}-8\,{c}^{2}xta+2\,{\rm i}{c}^{2}x
-8\,{c}^{2}{t}^{2}{a}^{2}+8\,{c}^{2}xt{a}^{2}.
\end{align*}

 As {}{it is} known, there exist two holes near the peak in the first-order RW.
These two holes are located at $P_1=({\frac {-18a+12}{\sqrt {-24\,a+24+3\,{c}^{2}} \left( 2\,a-2-{c}^{2}
 \right) c}},\,{\frac {3}{\sqrt {-24\,a+24+3\,{c}^{2}} \left( 2\,a-2-{c}^{2}
 \right) c}})$ and $P_2=({\frac {18a-12}{\sqrt {-24\,a+24+3\,{c}^{2}} \left( 2\,a-2-{c}^{2}
 \right) c}},$ $\,{\frac {-3}{\sqrt {-24\,a+24+3\,{c}^{2}} \left( 2\,a-2-{c}^{2}
 \right) c}})$.  It is obvious that $P_1$ and $P_2$ are on the line $l_1:\,x=-2(3a-2)t$. On the background plane
  with {}{height} $|r_{2r}|=c$,  the contour line is a hyperbola
\begin{align}\label{hyperbola}
  &\left( 4\,{c}^{6}-4\,{c}^{4}-8\,{c}^{2}{a}^{3}+3\,{c}^{8}+24\,a{c}^{4
}-16\,{a}^{2}{c}^{4}+8\,{c}^{2}{a}^{2}+4\,a{c}^{6} \right) {t}^{2}+
 \left( -4\,a{c}^{4}+4\,{c}^{6}+8\,{c}^{4}+8\,{c}^{2}a\right. \nonumber \\
 &\left.-8\,{c}^{2}{a}^{2} \right) xt-1+ \left( {c}^{4}+2\,{c}^{2}-2\,{c}^{2}a \right) {x}^{2}=0,
\end{align}
which intersects with the line $l_1$ at two points $P_3=({\frac {6a-4}{\sqrt {-8\,a+3\,{c}^{2}+8} \left( 2\,a-2-{c}^{2}
\right) c}},\,-{\frac {1}{\sqrt {-8\,a+3\,{c}^{2}+8}
\left( 2\,a-2-{c}^{2} \right) c
}})$ and $P_4=({-\frac {6a-4}{\sqrt {-8\,a+3\,{c}^{2}+8} \left( 2\,a-2-{c}^{2} \right) c}},\,{\frac {1}
{\sqrt {-8\,a+3\,{c}^{2}+8} \left( 2\,a-2-{c}^{2} \right)c}})$. We define the tangential direction of
hyperbola to be at two points $P_3$ and $P_4$, which is the length-direction, {}{as described  by a line $l_2$ : $x=-(2a+3/2c^2)t$}. The {}{density  plot for $|r_{2r}|^2$} {}{combined} with the hyperbola and the length-direction is displayed in Fig. \ref{contour}.

\begin{figure}[!htbp]
\centering
\raisebox{20 ex}{(a)}\subfigure{\includegraphics[height=6cm,width=6cm]{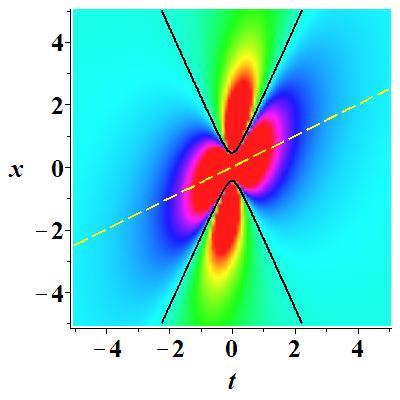}}
\qquad
\raisebox{20 ex}{(b)}\subfigure{\includegraphics[height=6cm,width=6cm]{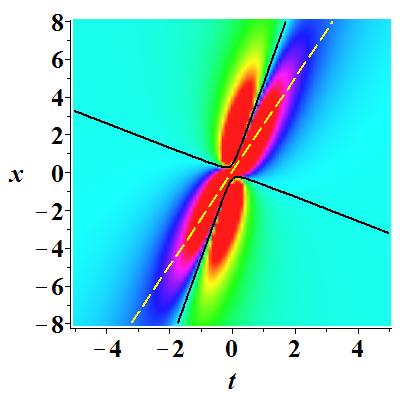}}
\caption{(Color online) The density plots of the first-order RW $|r_{2r}|^2$ with hyperbola($|r_{2r}|^2=c^2$) and length-direction. (a) $a=-1,\, c=1$. (b) $a=-2,\,c=1$. The black solid line is hyperbola, the yellow dash line is the length-direction.}\label{contour}
\end{figure}

Since the contour line is not closed on the background in the length-direction,
we have to {}{select a contour $|r_{2r}|^2-2c^2=0$ with height twice the background such that it is closed}. The closed contour is useful to discuss the localization characters of the the first-order RW. {}{It} intersects with the length-direction at two points. We define the distance $d^L$ of these two points as the length of the first-order RW, and we determine the projection $d^W$ of $|P_1P_2|$ on the width-direction, which is perpendicular to the length-direction,
{}{to be} the width of the first-order RW. Through a simple calculation, we {}{obtain}
\begin{equation}\label{length1}
  d^L=\frac{{d^L}_n}{{d^L}_d},
\end{equation}
with
\begin{align*}
{d^L}_n=&2\,\sqrt{\left(4+16\,{a}^{2}+24\,{c}^{2}a+
9\,{c}^{4}\right)\left( 48+9\,{c}^{4}+48\,{a}^{2}+46\,{c}^{2}-96\,a-46\,{
c}^{2}a+2\,\sqrt {M_{{1}}} \right)},\\
{d^L}_d=&( 8\,a-8-{c}^{2})( 2\,a-2-{c}^{2})c^2,\\
M_1=&1024-4096\,a+22\,{c}^{8}-4992\,{c}^{2}a+1024\,{a}^{4}+242\,{c}^{6}-
4096\,{a}^{3}-1664\,{c}^{2}{a}^{3}-242\,a{c}^{6}\\
&+976\,{a}^{2}{c}^{4}-1952\,a{c}^{4}+4992\,{c}^{2}{a}^{2}+1664\,{c}^{2}+6144\,{a}^{2}+976\,{
c}^{4},
\end{align*}
and
\begin{equation}\label{width1}
  d^W=\frac{{d^W}_n}{{d^W}_d},
\end{equation}
while
\begin{align*}
  {d^W}_n=6(8a-3c^2-8), \quad
  {d^W}_d=\sqrt{\left(-24a+24+3c^2\right)\left(4+16a^2+24c^2+9c^4\right)}\left(2a-2-c^2\right)c.
\end{align*}
The length $d^L$ and width $d^W$ are related to $a$ and $c$, and their profiles are plotted in Fig. \ref{dLW}.

\begin{figure}[!htbp]
\centering
\raisebox{20 ex}{(a)}\subfigure{\includegraphics[height=6cm,width=6cm]{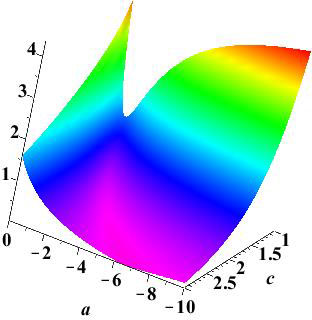}}
\qquad
\raisebox{20 ex}{(b)}\subfigure{\includegraphics[height=6cm,width=6cm]{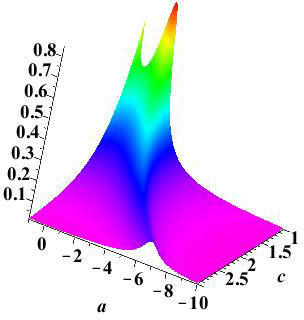}}
\caption{(Color online) {}{The localization characters of the first-order RW of the CLL-NLS, which are two functions of $a$ and $c$.}
(a) The length $d^L$. (b)The width $d^W$. }\label{dLW}
\end{figure}

 If one fixes the parameter $c$, the length decreases with the increase of $a$ at first and then increases
 until $a=1$. On the other hand, the width increases first, decreases, and then increases {}{again} until $a=1$. For example,
 when $c=1$, the length decreases with $a$ if $a\in(-\infty,-0.88)$ and increases with $a$ if $a\in(-0.88,1)$. At the same
 time, the width increases with $a$ if $a\in(-\infty,-0.69)$ or $a\in(0.52,1)$ and decreases with $a$ if $a\in(-0.69,0.52)$.
 Furthermore, when $a$ tends to $-\infty$, $d^L$ tends to $2\sqrt7$ and $d^W$ tends to $0$, $d^L$ reaches to the minimum
  $1.32$ when $a=-0.88$ and gets to the maximum $62.42$ when $a\rightarrow1$,  and $d^W$ reaches to the maximum $1.70$
  when $a=-0.69$.  In order to provide a visual support of above analysis on the trend {}{with respect to} $a$ of two localization characters of the  RW for the CLL-NLS, two curves for $d^L$ and $d^W$  with fixed $c=1$ are given in Fig. \ref{ld2d}(a).

In order to consider the {}{contribution} of the SSE on the localization characters of the RW, we define the length and width
 of the RW as mentioned above for the NLS {}{${\rm i}r_t+r_{xx}+|r|^2r=0$, which is {}{trivially} given by ignoring the SSE term in the
 CLL-NLS}. After applying a scaling transformation, due to the the different coefficient of nonlinear term, the first-order RW of the NLS
 can be obtained from the results{}{, reported in} \cite{He1}. Then, the length and width of the first-order RW of the NLS are expressed by
\begin{equation}
  {d^L}_{NLS}=\frac{\sqrt{7(1+4a^2)}}{2c^2},\qquad {d^W}_{NLS}=\frac{\sqrt{3}}{(1+4a^2)c},
\end{equation}
and the length direction is described by a line $l_{2NLS}$: $x=2at$.

Set $c=1$, then ${d^L}_{NLS}$ and ${d^W}_{NLS}$ reach the minimum $\frac{\sqrt7}2$ and the maximum $\sqrt3$ at $a=0$, respectively.
It implies that the maximum of width and the minimum of length of the RW for the NLS are roughly equal to the corresponding
values of the RW for the CLL-NLS. The width of the RW for the NLS also tends to $0$ when $a\rightarrow\pm\infty$.
 However, the length tends to $+\infty$, when $a\rightarrow\pm\infty$. There is no oscillation interval for the width of the RW solution
  of the NLS. This is different from the analogous CLL-NLS. The profiles of ${d^L}_{NLS}$ and ${d^W}_{NLS}$ with $c=1$
  are given in Fig. \ref{ld2d}(b).
	
\begin{figure}[!htbp]
\centering
\raisebox{20 ex}{(a)}\subfigure{\includegraphics[height=6cm,width=6cm]{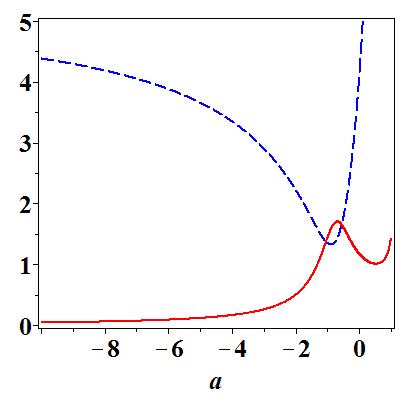}}
\qquad
\raisebox{20 ex}{(b)}\subfigure{\includegraphics[height=6cm,width=6cm]{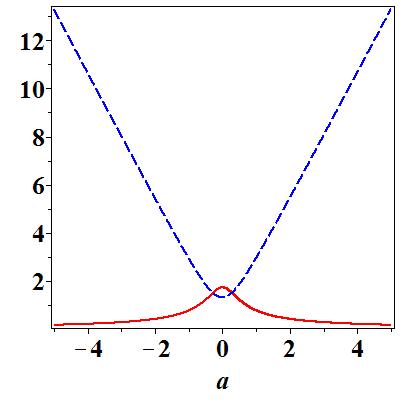}}
\caption{(Color online) The length $d^L$ and the width $d^W$ for the first-order RW with c=1. The blue dash line indicates the length,
 and the red solid line denotes the width. (a) The CLL-NLS. (b) The NLS.}\label{ld2d}
\end{figure}
	
Furthermore, we notice that $d^L<{d^L}_{NLS}$ if $a\in(-\infty,\, -0.47)$ and $d^L>{d^L}_{NLS}$ if $a\in(-0.47,1)$, $d^W<{d^W}_{NLS}$, if
$a\in(-\infty,\, -2.53)$ or $a\in(-0.33,\,0.67)$, and $d^W>{d^W}_{NLS}$ if $a\in(-2.53,\, -0.33)$ or $a\in(0.67,\,1)$
in the case of $c=1$. These detailed comparisons on localization characters of the first-order RWs are
given in table 1 and Fig.\ref{wdcomparsion}.

\begin{figure}[!htbp]
\centering
\raisebox{20 ex}{(a)}\subfigure{\includegraphics[height=6cm,width=6cm]{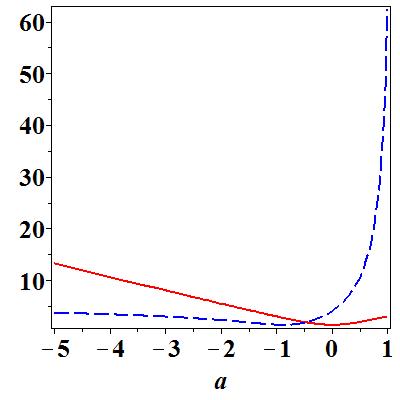}}
\qquad
\raisebox{20 ex}{(b)}\subfigure{\includegraphics[height=6cm,width=6cm]{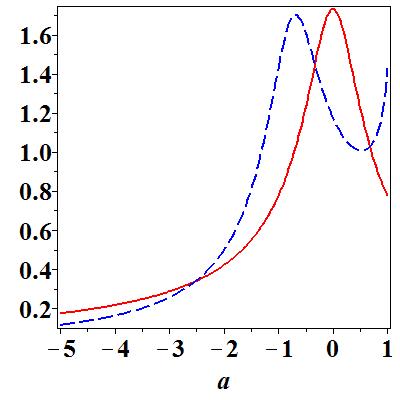}}
\caption{(Color online) The comparison between length (a) and width (b) of the first-order  RWs
for the NLS  equation(red,solid) and the CLL-NLS(blue, dash) with $a<1$ and $c=1$.
The left panel has one intersection point ($-0.47, 1.82$).
There are three  intersection points in the right panel: ($-2.53, 0.34$), ($-0.33,1.45$),($0.67,1.03$). }\label{wdcomparsion}
\end{figure}

 This analysis is visually verified by contours of $|r|^2$ for heights, being twice higher than the background in Fig.\ref{ld3d}. Furthermore, since the length and the width of the first-order RW of CLL-NLS are smaller than the corresponding NLS case when $a<-2.53$, respectively, the latter is therefore \textit{better} than the corresponding NLS one. {}{From an experimental point of view, having a smaller localization, we emphasize therefore a simpler set-up, since the proapgating distance of first-order CLL-NLS RW is considerably smaller compared to the NLS case.} The opposite case is valid for $-0.47<a<-0.33$ and $0.67< a<1$, where the CLL-NLS RW is \textit{worse}. Unfortunately, we have not been able to compare the localization properties, when $a$ belongs to one of the other two intervals, shown in the third and fifth column of table $1$. This is due to the fact that the width and length of the corresponding localization is alternatively smaller or bigger for the CLL-NLS compared to the NLS. In other words, the SSE in the CLL-NLS gives a remarkable change of the localization properties of the first-order RW, although we are not able to claim, if the RW localization for this equation is rather improved or destroyed by this term at different points $(a,c)$, in the parameter space. This is the first impact of the SSE on RW solutions of the CLL-NLS. As second impact we emphasize is that the SSE induces a strong rotation of the direction length on the RW of the CLL-NLS by comparing the two lines $l_2$ and $l_{2NLS}$. These two impacts are demonstrated intuitively by contours at a height $2c^2$ of the modulus square for the first-order RWs of the CLL-NLS and of the NLS in Fig. \ref{ld3d}.

\begin{figure}[!htbp]
\centering
\raisebox{20 ex}{(a)}\subfigure{\includegraphics[height=4.5cm,width=4.5cm]{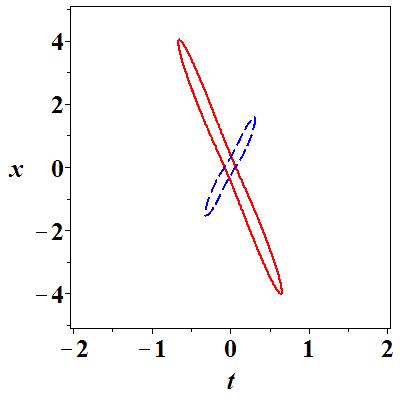}}
\qquad
\raisebox{20 ex}{(b)}\subfigure{\includegraphics[height=4.5cm,width=4.5cm]{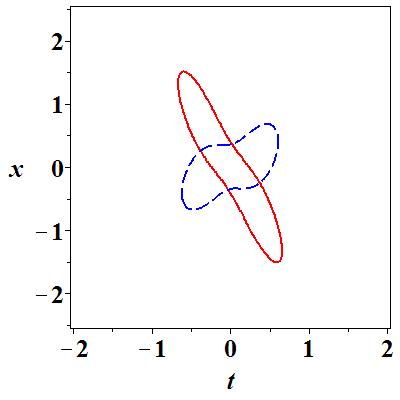}}
\\
\raisebox{20 ex}{(c)}\subfigure{\includegraphics[height=4.5cm,width=4.5cm]{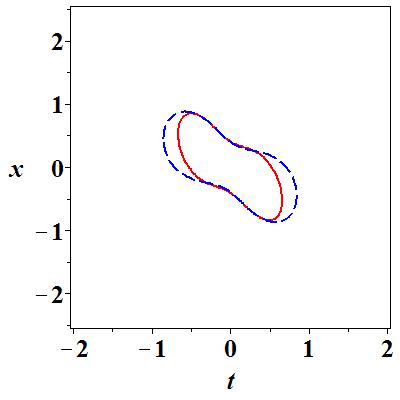}}
%\quad
\raisebox{20 ex}{(d)}\subfigure{\includegraphics[height=4.5cm,width=4.5cm]{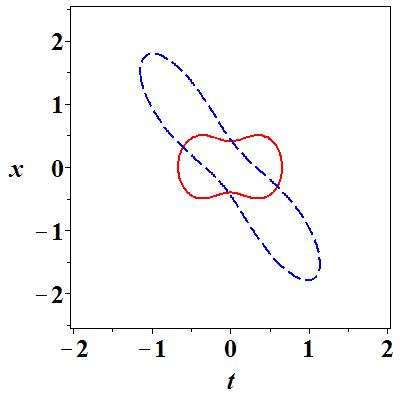}}
%\quad
\raisebox{20 ex}{(e)}\subfigure{\includegraphics[height=4.5cm,width=4.5cm]{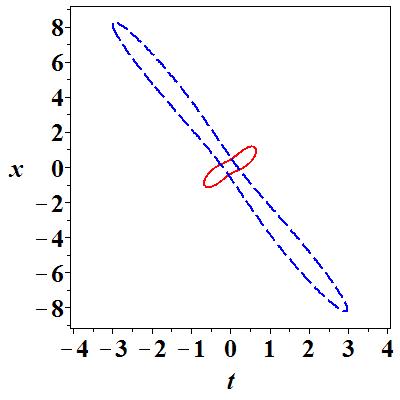}}
\caption{(Color online) The contours of first-order RWs  at height($2c^2$) twice background with $c=1$ and different values of $a$.
The red solid line indicates the NLS and the blue dash line denotes the CLL-NLS.
(a) $a=-3$, (b) $a=-1$, (c)$a=-0.4$, (d)$a=0$, (e)$a=0.7$.}\label{ld3d}
\end{figure}

\small{
\begin{table}[t!]%[h!b!p!]
\label{ttt1}
\begin{center}
  {The localization characters for the RW  in the  NLS and CLL-NLS}
\end{center}
\vspace{2 ex}
\centering
\begin{tabular}{|c|c|c|c|c|c|c|}
\hline
Values of $a$     & $a<-2.53$        & $-2.53<a<-0.47$          &$-0.47<a<-0.33$    &$0.33<a<0.67$ &$0.67<a<1$\\\hline
%%%%%%%%%%%%%%%
Length    & $d^L_{NLS}>d^L$   & $d^L_{NLS}>d^L$       &$d^L_{NLS}<d^L$    &$d^L_{NLS}<d^L$ &$d^L_{NLS}<d^L$\\\hline
Width      & $d^W_{NLS}>d^W$  &$d^W_{NLS}<d^W$        &$d^W_{NLS}<d^W$    &$d^W_{NLS}>d^W$ &$d^W_{NLS}{}{<}d^W$\\\hline
%%%%%%%%%%%%%%%%%%%%
Localization & NLS$<$CLL-NLS    & {}{Indeterminate}         &NLS$>$CLL-NLS        &{}{Indeterminate}  &NLS$>$CLL-NLS\\
\hline
\end{tabular}
\vspace{2 ex}
\caption{NLS$<$CLL-NLS means the localization
of RW in the CLL-NLS is \textit{better} than NLS, since the width and length of the CLL-NLS RW is smaller compared to the NLS, respectively. NLS$>$CLL-NLS is the opposite case.}
\end{table}}

\section{Conclusions}
We have shown that exact fundamental and higher-order solution of the CLL-NLS can be constructed, using the DT. These solutions {}{may} describe the accurate propagation of localized structures in nonlinear dispersive media, since dispersion, nonlinearity and SSE have been taken into account. In particular, we provide exact analytical expressions for doubly-localized RW solutions. Furthermore, we discuss the influence of the SSE on the {}{localization characteristics} of NLS RWs using visualization contour method. This work may motivate similar studies for higher-order evolution equation of this kind, such for higher-order generalized nonlinear Schr\"odinger-type equations. In particular, experiments in several nonlinear dispersive media, such in nonlinear optical fibers or in water wave flumes may be a consequence of these studies.

\mbox{\vspace{1cm}}

%%%%%%%%%%%%%%%%%%%%%%%%%%%%%%%%%%%%%%%%%%%%%%%%%%%%%%%%%%%%%%%%%%%%%%%%%%%%%%%%
%%%%%%%%%%%%%%%%%%%%%%%%%%%%%%%%%%%%%%%%%%%%%%%%%%%%%%%%%%%%%%%%%%%%%%%%%%%%%%%%5
%%%%%%%%%%%%%%%%%%%%%%%%%%%%%%%%%%%%%%%%%%%%%%%%%%%%%%%%%%%%5

{\bf Acknowledgments}  This work is supported by the NSF of China under Grant No.11271210, the K. C. Wong Magna Fund in Ningbo University. J. S. H acknowledges sincerely Prof. A. S. Fokas for arranging the visit to Cambridge University in 2012-2014 and for many useful discussions. A. C. acknowledges support from the Isaac Newton Institute for Mathematical Sciences.

%%%%%%%%%%%%%%%%%%%%%%%%%%%%%%%%%%%%%%%%%%%

\end{document}